\documentclass{article}%
\usepackage{amsfonts}
\usepackage{amsmath}%
\usepackage{amssymb}%
\usepackage{graphicx}
\providecommand{\U}[1]{\protect\rule{.1in}{.1in}}
\newtheorem{theorem}{Theorem}

\newtheorem{corollary}[theorem]{Corollary}

\newtheorem{definition}[theorem]{Definition}
\newtheorem{example}[theorem]{Example}

\newtheorem{lemma}[theorem]{Lemma}

\newtheorem{proposition}[theorem]{Proposition}
\newtheorem{remark}[theorem]{Remark}

\newenvironment{proof}[1][Proof]{\noindent\textbf{#1.} }{\ \rule{0.5em}{0.5em}}
\begin{document}

\title{\textbf{Equivariant differential characters and Chern-Simons bundles}}
\date{}
\author{\textsc{Roberto Ferreiro P\'{e}rez}\\Departamento de Econom\'{\i}a Financiera y Actuarial y Estad\'{\i}stica\\Facultad de Ciencias Econ\'omicas y Empresariales, UCM\\Campus de Somosaguas, 28223-Pozuelo de Alarc\'on, Spain\\\emph{E-mail:} \texttt{roferreiro@ccee.ucm.es}}
\maketitle

\begin{abstract}
We construct Chern-Simons bundles as $\mathrm{Aut}^{+}P$-equivariant
$U(1)$-bundles with connection over the space of connections $\mathcal{A}_{P}%
$\ on a principal $G$-bundle $P\rightarrow M$. We show that the Chern-Simons
bundles are determined up to isomorphisms by their equivariant holonomy. The
space of equivariant holonomies is shown to coincide with the space of
equivariant differential characteres of order $2$. Furthermore, we prove that
the Chern-Simons theory provides, in a natural way, an equivariant
differential character that determines the Chern-Simons bundles. Our
construction can be applied in the case in which $M$ is a compact manifold of
even dimension and for arbitrary bundle \thinspace$P$ and group $G$.

The results are also generalized to the case of the action of diffeomorphisms
on the space of Riemannian metrics. In particular, in dimension $2$ a
Chern-Simons bundle over the Teichm\"{u}ller space is obtained.

\end{abstract}


\smallskip


\medskip

\noindent\emph{Acknowledgments:\/} Supported by Ministerio de Ciencia,
Innovaci\'{o}n y Universidades of Spain under grant PGC2018-098321-B-I00.

\section{Introduction}

We introduce a geometric definition of Chern-Simons bundles valid for
arbitrary Lie groups and principal bundles over even dimensional compact
manifolds. Our definition is based on the concept of equivariant holonomy
introduced in \cite{AnomaliesG} and \cite{EquiHol}. We show that the
Chern-Simons bundles can be obtained from its equivariant holonomy, and that
the equivariant holonomy is determined in a natural way from the
Cheeger-Chern-Simons differential characters introduced in \cite{Cheeger}.

First we recall the classical construction of the Chern-Simons bundles. If $M
$ is a closed $2$-manifold, then the space of connections $\mathcal{A}_{P}$ on
the trivial principal $SU(2)$-bundle $P=M\times SU(2)\rightarrow M$ is a
symplectic manifold with the Atiyah-Bott symplectic structure. It is a
classical construction in Chern-Simons theory (e.g. see \cite{RSW}) that this
symplectic manifold admits a $\mathrm{Gau}(P)$-equivariant prequantization
bundle $\mathcal{U}_{P}\rightarrow\mathcal{A}_{P}$ with connection $\Theta
_{P}$. By symplectic reduction, a prequantization of the Atiyah-Bott
symplectic structure on the moduli space of flat connections is obtained.
Furthermore, if $\widetilde{M}$ is a compact 3-manifold with boundary
$\partial\widetilde{M}=M$, $\widetilde{P}=\widetilde{M}\times SU(2) $ and
$r\colon\mathcal{A}_{\widetilde{P}}\rightarrow\mathcal{A}_{P}$ is the
restriction map, then the Chern-Simons action on $\widetilde{M}$ can be
considered as a section of the bundle $r^{\ast}\mathcal{U}_{P}^{-1}%
\rightarrow\mathcal{A}_{\widetilde{P}}$ (see also \cite{Freed1}). For this
reason the bundle with connection $(\mathcal{U}_{P},\Theta_{P})$ is called the
Chern-Simons bundle of $M$. We also recall that it is possible (e.g. see
\cite{Andersen}) to lift the action of the group of orientation preserving
diffeomorphisms $\mathcal{D}_{M}^{+}$ to $\mathcal{U}_{P}$ preserving the
connection $\Theta_{P}$.

The construction of the Chern-Simons bundle in \cite{RSW} can be easily
extended to trivial bundles with arbitrary group $G$ (e.g. see
Section\ \ref{SectEquiHolCS}). In this case the bundle is constructed by using
the Chern-Simons action associated to a Weil polynomial of $G$. For connected
and simply connected group $G$ any bundle over a $2$ or $3$-manifold is
trivial, and hence the preceding construction can be applied. However, for
nontrivial bundles and in higher dimensions this construction cannot be applied.

It is shown in \cite{CSconnections} that it is possible to define the
Chern-Simons bundle for an arbitrary principal $G$-bundle $P\rightarrow M$
with base a closed manifold $M$ of dimension $2k-2$, $k\geq2$. The
Chern-Simons bundle is associated to a Weil polynomial $p\in I^{k}(G)$, a
compatible characteristic class $\Upsilon\in H^{2k}(\mathbf{B}G)$ and a
background connection $A_{0}$. Moreover, the bundle is equivariant with
respect to the action of the group $\mathrm{Aut}^{+}P$ of automorphisms of $P
$ (and not only under gauge transformations as in \cite{RSW}). It is also
proved in \cite{CSconnections} that the bundles associated to different
background connections $A_{0}$ are canonically isomorphic. The construction of
the Chern-Simons bundle in \cite{CSconnections}\ is rather technical and hard
to interpret in geometrical terms. As commented above, in this paper we
clarify the construction of Chern-Simons bundle by using the concept of
equivariant holonomy that we recall (see \cite{EquiHol} for details). Let a
Lie group $\mathcal{G}$ act on a manifold $N$. The ordinary holonomy is
defined for closed curves. The equivariant analogue of a closed curve is a
curve $\gamma$ such that $\gamma(1)=\phi_{N}(\gamma(0))$ for an element
$\phi\in\mathcal{G}$. Note that in this case, if $\pi\colon N\rightarrow
N/\mathcal{G}$ is the projection then $\pi\circ\gamma$ is a closed curve on
$N/\mathcal{G}$. Let $\mathcal{U}\rightarrow N$ be a $\mathcal{G}$-equivariant
principal $U(1)$-bundle with a $\mathcal{G}$-invariant connection $\Theta$. We
define the $\phi$-equivariant log-holonomy $\mathrm{hol}_{\phi}^{\Theta
}(\gamma)\in\mathbb{R}/\mathbb{Z}$ of $\gamma$ by the property $\overline
{\gamma}(1)=\phi_{\mathcal{U}}(\overline{\gamma}(0))\cdot\exp(2\pi
i\mathrm{hol}_{\phi}^{\Theta}(\gamma))$ for any $\Theta$-horizontal lift
$\overline{\gamma}\colon\lbrack0,1]\rightarrow\mathcal{U}$ of $\gamma$.
Moreover, the equivariant holonomy determines the $U(1)$-bundle with
connection up to $\mathcal{G}$-equivariant isomorphisms (see \cite{EquiHol}).

In the case of the Chern-Simons bundle $(\mathcal{U}_{P,}\Theta_{P})$ over the
space of connections of a trivial $SU(2)$-bundle it is possible to\ compute
the equivariant holonomy of $\Theta_{P}$ (see Section \ref{SectEquiHolCS}) and
the result is as follows. If $\phi\in\mathrm{Aut}P$, then a curve $\gamma$ on
$\mathcal{A}_{P}$ determines a connection $A^{\gamma}$ on $P\times
I\rightarrow M\times I$ such that $A^{\gamma}|_{P\times\{t\}}=\gamma(t)$. We
denote by $P_{\phi}=(P\times I)/\sim_{\phi}$ the mapping torus of $P$, where
$(y,0)\sim_{\phi}(\phi_{P}(y),1) $ for any $y\in P$. The condition
$\gamma(1)=\phi_{\mathcal{A}_{P}}(\gamma(0))$ implies that $A^{\gamma}$
projects onto a connection $A_{\phi}^{\gamma}$ on $P_{\phi}\rightarrow
M_{\phi}$. We prove in Section \ref{SectEquiHolCS}\ that if $\phi
\in\mathrm{Gau}P$, then we have$\ $%
\begin{equation}
\mathrm{hol}_{\phi}^{\Theta_{P}}(\gamma)=-\mathrm{CS}(A_{\phi}^{\gamma
})\label{holCS}%
\end{equation}
where $\mathrm{CS}\colon\mathcal{A}_{P_{\phi}}\rightarrow\mathbb{R}%
/\mathbb{Z}$ is the usual Chern-Simons action ($\mathrm{CS}$ is well defined
because $P_{\phi}\rightarrow M\times S^{1}$ is a principal $SU(2)$-bundle over
a $3$-manifold, and hence trivializable).

As the equivariant holonomy determines the equivariant bundle up to an
isomorphism, we can use equation (\ref{holCS}) to define the Chern-Simons
bundle for arbitrary bundles and dimensions. We recall (e.g. see \cite{DW})
that the Chern-Simons action can be extended to arbitrary bundles. We define a
characteristic pair as a pair $\vec{p}=(p,\Upsilon)$, where $p\in
I_{\mathbb{Z}}^{r}(G)$ is Weil polynomial\ and $\Upsilon\in H^{2r}%
(\mathbf{B}G,\mathbb{Z})$ a compatible characteristic class. Given the
characteristic pair $\vec{p}$, for any principal $G$-bundle $Q\rightarrow N$
over a closed manifold $N$ of dimension $2r-1$ the Chern-Simons action is
defined $\mathrm{CS}^{\vec{p}}\colon\mathcal{A}_{Q}\rightarrow\mathbb{R}%
/\mathbb{Z}$ (see Section \ref{CCS}). Hence, if $M$ is a closed manifold of
dimension $2r-2$ and $P\rightarrow M$ is a principal $G$-bundle, then for any
$\phi\in\mathrm{Gau}P$ we can define
\begin{equation}
\Xi_{P}^{\vec{p}}(\phi,\gamma)=\mathrm{CS}^{\vec{p}}(A_{\phi}^{\gamma
})\label{CSChar}%
\end{equation}
and we know that there is at most one (up to an isomorphism) $\mathrm{Gau}%
P$-equivariant $U(1)$-bundle with connection whose $\mathrm{Gau}P$-equivariant
holonomy is given by $\Xi_{P}^{\vec{p}}$. Furthermore, we can compute
(\ref{CSChar}) also for any $\phi\in\mathrm{Aut}^{+}P$.

Now the remaining question is if the Chern-Simons bundle exists, i.e., if
there exists a $\mathrm{Aut}^{+}P$-equivariant $U(1)$-bundle with connection
$(\mathcal{U}_{P}^{\vec{p}},\Theta_{P}^{\vec{p}})$ such that $\mathrm{hol}%
_{\phi}^{\Theta_{P}^{\vec{p}}}(\gamma)=\Xi_{P}^{\vec{p}}(\phi,\gamma)$. We
prove that it exists by introducing the concept of equivariant differential character.

We recall that in the non-equivariant setting, the set of log-holonomies of
connections on principal $U(1)$-bundles over $N$ is known to coincide with the
space $\hat{H}^{2}(N)$ of Cheeger-Simons differential characters of degree $2$
(see Section \ref{Other def} for details). Furthermore, it is a classical
result that if $\mathrm{Bund}_{U(1)}^{\nabla}(N)$ denotes the set of principal
$U(1)$-bundles with connection over a manifold $N$ modulo isomorphisms, then
the map that assigns to a connection its holonomy induces a bijection
\begin{equation}
\mathrm{Bund}_{U(1)}^{\nabla}(N)\simeq\hat{H}^{2}(N).\label{iso}%
\end{equation}

In Section \ref{EDC}\ we define the space of equivariant differential
characters $\hat{H}_{\mathcal{G}}^{2}(N)$ and we prove an equivariant version
of the isomorphism (\ref{iso}) in the case in which $N$ is contractible (see
Section \ref{Other def} for the case of arbitrary $N$).

Finally, in the case of the Chern-Simons bundle, we prove in Section
\ref{SectEquiCS}\ that we have $\Xi_{P}^{\vec{p}}\in\hat{H}_{\mathrm{Aut}%
^{+}P}^{2}(\mathcal{A}_{P})$, where $\Xi_{P}^{\vec{p}}$ is defined by equation
(\ref{CSChar}). We conclude that the Chern-Simons bundle exists as a
$\mathrm{Aut}^{+}P$-equivariant bundle for any principal $G$-bundle
$P\rightarrow M$. It is important to note that the equivariant differential
character $\Xi_{P}^{\vec{p}}$ is determined only by $\vec{p}$, but to obtain a
concrete bundle and connection additional information is necessary. It is
shown in Section \ref{SectEquiCS} that if we fix a background connection
$A_{0}\in\mathcal{A}_{P}$, it is possible to obtain a concrete $\mathrm{Aut}%
^{+}P$-equivariant bundle with connection. In this way we recover the result
proved in \cite{CSconnections}, but our proofs are simpler and more conceptual.

\subsection{Further applications}

The equivariant differential character $\Xi_{P}^{\vec{p}}$ is the fundamental
geometrical object from which other geometric constructions can be derived.
For example, let $\mathcal{F}_{P}\subset\mathcal{A}_{P}$ be the space of flat
connections and let $\mathcal{G}$ be a subgroup of $\mathrm{Aut}^{+}P$ acting
freely on $\mathcal{F}_{P}$. Then $\Xi_{P}^{\vec{p}}$ projects onto a
differential character $\underline{\Xi}_{P}^{\vec{p},\mathcal{F}}\in\hat
{H}_{\mathrm{Aut}^{+}P}^{2}(\mathcal{F}_{P}/\mathcal{G})$, and hence
determines a $U(1)$-bundle $\underline{\mathcal{U}}^{\vec{p}}\rightarrow
\mathcal{F}_{P}/\mathcal{G}$ with connection $\underline{\Theta}^{\vec{p}}$
(defined up to an isomorphism). In the case in which $G=SU(2)$, $\vec{p}$ is
the characteristic pair corresponding to the second Chern class, $M$ is a
Riemann surface and $P=M\times SU(2)$ the bundle $(\underline{\mathcal{U}%
}^{\vec{p}},\underline{\Theta}^{\vec{p}})$ is isomorphic to Quillen's
determinant line bundle. In Section \ref{SectGau} we study the restriction of
$\Xi_{P}^{\vec{p}}$ to the action of the Gauge group, and\ it is\ shown how
the classical constructions in Chern-Simons theory can be generalized to
arbitrary bundles. In Section \ref{Sec Metrics} we consider the action of the
orientation preserving diffeomorphisms group $\mathcal{D}_{M}^{+}$ on the
space of Riemannian metrics $\mathcal{M}_{M}$ on a compact oriented manifold
$M$\ of dimension $4k-2$. Precisely, let $FM\rightarrow M$ be the frame bundle
of $M$ and let $\vec{p}$ be the characteristic pair corresponding to the
$k$-th Pontryagin class. Then we can pull-back the character $\Xi_{FM}%
^{\vec{p}}$ by the Levi-Civita map and we obtain a equivariant differential
character $\Sigma_{\mathcal{M}_{M}}^{\vec{p}}\in\hat{H}_{\mathcal{D}_{M}^{+}%
}^{2}(\mathcal{M}_{M})$. In the case $k=1$, if $M$ is a Riemann surface of
genus $g>1$, then $\Sigma_{\mathcal{M}_{M}}^{\vec{p}}$ projects onto a
equivariant differential character $\underline{\Sigma}_{\mathcal{M}_{M}}%
^{\vec{p}}\in\hat{H}_{\Gamma_{M}}^{2}(\mathcal{T}_{M})$, where $\mathcal{T}%
_{M}$ is the Teichm\"{u}ller space of $M $ and $\Gamma_{M}$ is the mapping
class group of $M$. Furthermore, the curvature of $\underline{\Sigma
}_{\mathcal{M}_{M}}^{\vec{p}}$ is $\frac{1}{2\pi}\sigma_{\mathrm{WP}},$ where
$\sigma_{\mathrm{WP}}$ is the symplectic form associated to the Weil-Petersson
metric on $\mathcal{T}_{M}$. By the the equivariant version of isomorphisms
(\ref{iso}) we obtain a $\Gamma_{M}$-equivariant prequantization bundle for
$\frac{1}{2\pi}\sigma_{\mathrm{WP}}$ (determined up to an isomorphism).

Furthermore, there are other important constructions in gauge theory that can
be interpreted as equivariant differential characters. One important\ example
is Witten global gravitational anomaly formula \cite{WittenGlobAn}. We recall
that in \cite{WittenGlobAn}, in order to study global gravitational anomalies,
Witten studies the variation of the path integral\footnote{$Z$ is defined as
the regularized determinant of a $\mathcal{D}_{M} $-equivariant family of
Dirac operators $\delta_{g}$ parametrized by Riemannian metrics $g\in
\mathcal{M}_{M}$} $Z$ under the action of the diffeomorphisms group
$\mathcal{D}_{M}$. He defines a number $w(\phi,\gamma)\in\mathbb{R}%
/\mathbb{Z}$ that measures the variation of $Z$ along a curve $\gamma\colon
I\rightarrow\mathcal{M}_{M}$ such that $\gamma(1)=\phi(\gamma(0))$ for a
diffeomorphism $\phi\in\mathcal{D}_{M}$. In more detail $w(\phi,\gamma
)=\lim\eta(\delta_{\phi})$, where $\eta$ denotes the Atiyah-Patodi-Singer
$\eta$-invariant, $\delta_{\phi}$ is an elliptic operator on the mapping torus
$M_{\phi}$ and $\lim$ denotes adiabatic limit\footnote{Note the similarity
between the definitions of $w$ and $\Xi_{P}^{\vec{p}}$.}. Later Witten's
formula was interpreted (e.g. see \cite{freed}) as a computation of the
holonomy of the Bismut-Freed connection $\Theta^{\delta} $ on the quotient
determinant line bundle $\det\delta/\mathcal{D}_{M}\rightarrow\mathcal{M}%
_{M}/\mathcal{D}_{M}$, or more precisely, as a computation of the
$\mathcal{D}_{M}$-equivariant holonomy on the equivariant determinant line
bundle $\det\delta\rightarrow\mathcal{M}_{M}$ (see \cite{LocUni},
\cite{EquiHol}). In particular $w\in\hat{H}_{\mathcal{D}_{M}}^{2}%
(\mathcal{M}_{M})$.

\section{Equivariant cohomology in the Cartan model}

First, we recall the definition of equivariant cohomology in the Cartan model
(\emph{e.g. }see \cite{GS}). Suppose that we have a left action of a connected
Lie group $G$ on a manifold $M$. If $\phi\in G$ and $x\in M$, we denote by
$\phi_{M}(x)$ or simply by $\phi\cdot x$ the action of $\phi$ on $x$. In a
similar way, for $X\in\mathfrak{g}$ the fundamental vector field $X_{M}%
\in\mathfrak{X}(M)$\ is defined by $X_{M}(x)=\left.  \frac{d}{dt}\right\vert
_{t=0}\exp(-tX)_{M}(x)$.

We denote by $\Omega^{k}(M)^{G}$ the space of $G$-invariant $k$-forms on $M$.
Let $\Omega_{G}^{\bullet}(M)=\mathcal{P}^{\bullet}(\mathfrak{g},\Omega
^{\bullet}(M))^{G}$ be the space of $G$-invariant
polynomials\footnote{Continuous polynomials in the infinite dimensional case.}
on $\mathfrak{g}$ with values in $\Omega^{\bullet}(M)$, with the graduation
$\deg(\alpha)=2k+r$ if $\alpha$ is a polynomial of degree $k$ with values on
the space $\Omega^{r}(M)$. Let $D\colon\Omega_{G}^{q}(M)\rightarrow\Omega
_{G}^{q+1}(M)$ be the Cartan differential, $(D\alpha)(X)=d(\alpha
(X))-\iota_{X_{M}}\alpha(X)$ for $X\in\mathfrak{g}$. On $\Omega_{G}^{\bullet
}(M)$ we have $D^{2}=0$, and the equivariant cohomology (in the Cartan model)
of $M$ with respect of the action of $G$ is defined as the cohomology of this
complex. If $\varpi\in\Omega_{G}^{2}(M)$ is a $G$-equivariant $2$-form, then
we have $\varpi=\omega+\mu$ where $\omega\in\Omega^{2}(M)$ is $G$-invariant
and $\mu\in\mathrm{Hom}\left(  \mathfrak{g},\Omega^{0}(M)\right)  ^{G}$. We
have $D\omega=0\;$if\ and only if $d\omega=0$, and $\iota_{X_{M}}\omega
=d(\mu_{X})\;$for every $X\in\mathfrak{g}$. Hence $\mu$ is a comoment map for
$\omega$.

Let a group $G$ act on a manifold $M$ and let $\rho\colon H\rightarrow G$ be a
homomorphism. We denote by $d\rho\colon\mathfrak{h}\rightarrow\mathfrak{g}$
the induced map on Lie algebras. If $H$ acts on another manifold $N$ we say
that $f\colon N\rightarrow M$ is $\rho$-equivariant if $f(\phi_{N}%
(x))=\rho(\phi)_{M}(f(x))$ for any $x\in N$ and $\phi\in H$. In this case, we
have a map $(f,\rho)^{\ast}\colon\Omega_{G}^{2}(M)\rightarrow\Omega_{H}%
^{2}(N)$ defined by $((f,\rho)^{\ast}\alpha)(X)=f^{\ast}(\alpha(d\rho(X)))$
for $X\in\mathfrak{h}$ and $\alpha\in\Omega_{G}^{2}(M)$.

We recall the definition of equivariant characteristic classes (see
\cite{BV2}). Let $H\ $be a group that acts on a principal $G$-bundle
$P\rightarrow M${\ }and let $A$ be a connection on $P$ invariant under the
action of $H$. It can be proved that for every $X\in\mathfrak{h}$
the{\ $\mathfrak{g}$-valued function $A(X_{P})$ is of adjoint type and defines
a section of the adjoint bundle} $v_{A}(X)\in\Omega^{0}(M,\mathrm{ad}P)$. We
denote by $I^{r}(G)$\ the space of Weil polynomials of degree $r$. For every
$p\in I^{r}(G)$ the $H$-equivariant characteristic form $p_{H}^{A}\in
\Omega_{H}^{2r}(M{)}$\ associated to $p$ and $A$, is defined{\ by} $p_{H}%
^{A}(X)=p(F_{A}-v_{A}(X))$ for every $X\in\mathfrak{h}$ and we have
$Dp_{H}^{A}=0$.

If $\alpha\in\Omega^{k}(M\times N)$ with $M$ compact and oriented we define
$\int\nolimits_{M}\alpha\in\Omega^{k-d}(N)$ by $\left(  \int\nolimits_{M}%
\alpha\right)  _{y}(X_{1},\ldots,X_{k-d})=\int\nolimits_{M}\iota_{X_{k-d}%
}\cdots\iota_{X_{1}}\alpha$ for $y\in N$, $X_{1},\ldots,X_{d}\in T_{y}N$. If
$k<d$ we define $\int\nolimits_{M}\alpha=0$. We have $\int\nolimits_{N}%
\int\nolimits_{M}\alpha=\int\nolimits_{M\times N}\alpha$ and Stokes theorem
$d\int\nolimits_{M}\alpha=\int\nolimits_{M}d\alpha-(-1)^{k-d}\int
\nolimits_{\partial M}\alpha$. Furthermore, if a group $G$ acts on $M$ and $N$
then the integration map is extended to equivariant differential forms $%
%TCIMACRO{\tint \nolimits_{M}}%
%BeginExpansion
{\textstyle\int\nolimits_{M}}
%EndExpansion
\colon\Omega_{G}^{k}(M\times N)\rightarrow\Omega_{G}^{k-d}(N)$ by setting
$\left(
%TCIMACRO{\tint \nolimits_{M}}%
%BeginExpansion
{\textstyle\int\nolimits_{M}}
%EndExpansion
\alpha\right)  (X)=%
%TCIMACRO{\tint \nolimits_{M}}%
%BeginExpansion
{\textstyle\int\nolimits_{M}}
%EndExpansion
(\alpha(X))$ for $X\in\mathfrak{g}$, and we have $D\left(  \int\nolimits_{M}%
\alpha\right)  =\int\nolimits_{M}D\alpha-(-1)^{k-d}\int\nolimits_{\partial
M}\alpha$.

\section{Equivariant holonomy}

In this section we recall the definition and properties of equivariant
holonomy introduced in \cite{AnomaliesG} for bundles with contractible base
and in \cite{EquiHol} for arbitrary bundles. Let $G$ be a Lie group with Lie
algebra $\mathfrak{g}$ and let $M$ be a connected and oriented manifold. A $G
$-equivariant $U(1)$-bundle is a principal $U(1)$-bundle $\mathcal{U}%
\rightarrow M$ in which $G$ acts (on the left) by principal bundle
automorphisms. We denote by $I$ the interval $[0,1]$. If $\gamma\colon
I\rightarrow M$ is a curve, we define the inverse curve $\overleftarrow
{\gamma}\colon I\rightarrow M$ by $\overleftarrow{\gamma}(t)=\gamma(1-t)$.
Moreover, if $\gamma_{1}$ and $\gamma_{2}\mathcal{\ }$are curves with
$\gamma_{1}(1)=\gamma_{2}(0)$ we define $\gamma_{1}\ast\gamma_{2}\colon
I\rightarrow\mathbb{R}$ by $\gamma_{1}\ast\gamma_{2}(t)=\gamma_{1}(2t)$ for
$t\in\lbrack0,1/2]$ and $\gamma_{1}\ast\gamma_{2}(t)=\gamma_{2}(2t-1)$ for
$t\in\lbrack1/2,1]$. For any $\phi\in G$ we define%
\[
\mathcal{C}^{\phi}(M)=\{\gamma\colon I\rightarrow M\text{ }|\text{ }%
\gamma\text{ is piecewise smooth and }\gamma(1)=\phi_{M}(\gamma(0))\},
\]
and $\mathcal{C}_{x}^{\phi}(M)=\{\gamma\in\mathcal{C}^{\phi}(M)$ $|$
$\gamma(0)=x\}$, $\mathcal{C}^{G}(M)=\{(\phi,\gamma)$ $|$ $\phi\in G$,
$\gamma\in\mathcal{C}^{\phi}(M)\}$. Note that if $e\in G$ is the identity
element, then $\mathcal{C}_{x}^{e}(M)=\mathcal{C}_{x}(M)$ is the space of
loops based at $x$. If $\phi\in G$ and $\gamma\in\mathcal{C}_{x}^{\phi
^{\prime}}(M)$ then we define $\phi\cdot\gamma\in\mathcal{C}_{\phi_{M}%
(x)}^{\phi\cdot\phi^{\prime}\cdot\phi^{-1}}(M)$ by $(\phi\cdot\gamma
)(t)=\phi_{M}(\gamma(t))$. We say that two curves $\gamma_{1}$,$\gamma_{2}$ on
$M$ differ by a reparametrization if $\gamma_{1}=\gamma_{2}\circ\sigma$, for a
piecewise smooth function $\sigma\colon I\rightarrow I\ $such that
$\sigma(0)=0$, $\sigma(1)=1$.

Let $\Theta$ be a $G$-invariant connection on a $G$-equivariant $U(1)$-bundle
$\mathcal{U}\rightarrow M$. If $\phi\in G$ and $\gamma\in\mathcal{C}^{\phi
}(M)$, the $\phi$-equivariant holonomy $\mathrm{Hol}_{\phi}^{\Theta}%
(\gamma)\in U(1)$ of $\gamma$ is characterized by the property $\overline
{\gamma}(1)=\phi_{\mathcal{U}}(\overline{\gamma}(0))\cdot\mathrm{Hol}_{\phi
}^{\Theta}(\gamma)$ for any $\Theta$-horizontal lift $\overline{\gamma}\colon
I\rightarrow\mathcal{U}$ of $\gamma$. We define the $\phi$-equivariant
log-holonomy $\mathrm{hol}_{\phi}^{\Theta}(\gamma)\in\mathbb{R}/\mathbb{Z}$ by
$\mathrm{Hol}_{\phi}^{\Theta}(\gamma)=\exp(2\pi i\mathrm{hol}_{\phi}^{\Theta
}(\gamma))$. Note that if $\gamma\in\mathcal{C}_{x}^{e}(M)$ is a loop on $M$,
then $\mathrm{Hol}_{e}^{\Theta}(\gamma)$ is the ordinary holonomy of $\gamma
$.\ Furthermore, if $\gamma_{1}$,$\gamma_{2}\in\mathcal{C}^{\phi}(M)$ differ
by a reparametrization then we have $\mathrm{hol}_{\phi}^{\Theta}(\gamma
_{1})=\mathrm{hol}_{\phi}^{\Theta}(\gamma_{2})$. The following results are
proved in \cite{EquiHol}.

\begin{proposition}
\label{PropHol}If $\mathcal{U}\rightarrow M$ is a $G$-equivariant principal
$U(1)$-bundle, and $\Theta$ is a $G$-invariant connection on $\mathcal{U}$,
then for any $\phi$, $\phi^{\prime}\in G$,\ $\gamma\in\mathcal{C}^{\phi}%
(M)$\ and $x\in M$ we have

a) $\mathrm{hol}_{\phi^{\prime}\cdot\phi\cdot(\phi^{\prime})^{-1}}^{\Theta
}(\phi^{\prime}\cdot\gamma)=\mathrm{hol}_{\phi}^{\Theta}(\gamma)$.

b) If $\gamma^{\prime}\in\mathcal{C}_{\gamma(1)}^{\phi^{\prime}}(M)$, then we
have $\mathrm{hol}_{\phi^{\prime}\cdot\phi}^{\Theta}(\gamma\ast\gamma^{\prime
})=\mathrm{hol}_{\phi}^{\Theta}(\gamma)+\mathrm{hol}_{\phi^{\prime}}^{\Theta
}(\gamma^{\prime}).$

c) $\mathrm{hol}_{\phi^{-1}}^{\Theta}(\overleftarrow{\gamma})=-\mathrm{hol}%
_{\phi}^{\Theta}(\gamma)$.

d) If $\zeta\colon I\rightarrow M$ is a curve on $M$ such that $\zeta
(0)=\gamma(0)$ then $\mathrm{hol}_{\phi}^{\Theta}(\overleftarrow{\zeta}%
\ast\gamma\ast(\phi\cdot\zeta))=\mathrm{hol}_{\phi}^{\Theta}(\gamma) $.

f) Let $\mathcal{U}^{\prime}\rightarrow M^{\prime}$ be another $G$-equivariant
$U(1)$-bundle with connection and $\Phi\colon\mathcal{U}^{\prime}%
\rightarrow\mathcal{U}$ be a $G$-equivariant $U(1)$-bundle morphism that
covers $\underline{\Phi}\colon M^{\prime}\rightarrow M$. The connection
$\Theta^{\prime}=\Phi^{\ast}\Theta$ is $G$-invariant and we have
$\mathrm{hol}_{\phi}^{\Theta^{\prime}}(\gamma^{\prime})=\mathrm{hol}_{\phi
}^{\Theta}(\underline{\Phi}\circ\gamma^{\prime})$ for any $\phi\in G$ and
$\gamma^{\prime}\in\mathcal{C}^{\phi}(M^{\prime})$.
\end{proposition}

If $\mathcal{U}\rightarrow M$, $\mathcal{U}^{\prime}\rightarrow M$\ are\ two
$G$-equivariant $U(1)$-bundles then we write that $\mathcal{U}^{\prime}%
\simeq_{G}\mathcal{U}$ if there exists a $G$-equivariant $U(1)$-bundle
isomorphism $\Phi\colon\mathcal{U}^{\prime}\rightarrow\mathcal{U}$ covering
the identity map of $M$. We say that $\mathcal{U}$ is a trivial $G$%
-equivariant $U(1)$-bundle if $\mathcal{U}\simeq_{G}M\times U(1)$ for an
action of $G$ on $M$ and where $G$ acts trivially on $U(1)$. A $G$-equivariant
$U(1)$-bundle with connection is a pair $(\mathcal{U},\Theta)$, where
$\mathcal{U}\rightarrow M$ is a $G$-equivariant $U(1)$-bundle and $\Theta$ is
a $G$-invariant connection on $\mathcal{U}$. We write that $(\mathcal{U}%
,\Theta)\simeq_{G}(\mathcal{U}^{\prime},\Theta^{\prime})$ if there exists a
$G$-equivariant $U(1)$-bundle isomorphism $\Phi\colon\mathcal{U}^{\prime
}\rightarrow\mathcal{U}$ covering the identity map of $M $ such that
\ $\Phi^{\ast}\Theta=\Theta^{\prime}$.

\begin{theorem}
\label{anomaly} Let $(\mathcal{U},\Theta)$ and $(\mathcal{U}^{\prime}%
,\Theta^{\prime})$ be $G$-equivariant $U(1)$-bundles with connection over $M$.

a) We have $(\mathcal{U},\Theta)\simeq_{G}(\mathcal{U}^{\prime},\Theta
^{\prime})$ if and only if $\mathrm{hol}_{\phi}^{\Theta}(\gamma)=\mathrm{hol}%
_{\phi}^{\Theta^{\prime}}(\gamma)$ for all $\phi\in G$, and $\gamma
\in\mathcal{C}^{\phi}(M)$.

b) The bundle $\mathcal{U}\rightarrow M$ is a trivial $G$-equivariant
$U(1)$-bundle if and only if there exists a $G$-invariant $1$-form $\beta
\in\Omega^{1}(M)^{G}$ such that $\mathrm{hol}_{\phi}^{\Theta}(\gamma
)=\int_{\gamma}\beta\operatorname{mod}\mathbb{Z}$ for any $\phi\in G$ and
$\gamma\in\mathcal{C}^{\phi}(M)$.
\end{theorem}

\subsection{Equivariant Curvature}

If $\Theta$ is a $G$-invariant connection on a principal $U(1)$ bundle
$\mathcal{U}\rightarrow M$ then $\frac{i}{2\pi}D(\Theta)$ is the pull-back of
a closed $G$-equivariant $2$-form $\mathrm{curv}_{G}(\Theta)\in\Omega_{G}%
^{2}(M)$ called the $G$-equivariant curvature of $\Theta$. If $X\in
\mathfrak{g}$ then we have $\mathrm{curv}_{G}(\Theta)(X)=\mathrm{curv}%
(\Theta)+\mu_{X}^{\Theta}$, where $\mu_{X}^{\Theta}=-\frac{i}{2\pi}%
\Theta(X_{\mathcal{U}})$ is called the momentum of $\Theta$. As it is well
known, for bundles with arbitrary group the curvature of $\Theta$\ measures
the infinitesimal holonomy. For $U(1)$-bundles we have a more precise result
that is a generalization of the classical Gauss-Bonnet Theorem for surfaces

\begin{proposition}
\label{HolonomyInt}If $\gamma\in\mathcal{C}^{e}(M)$ and $\gamma=\partial
\sigma$ for $\sigma\in C_{2}(M)$ then\ we have $\mathrm{hol}^{\Theta}%
(\gamma)=\int_{\sigma}\mathrm{curv}(\Theta)\operatorname{mod}\mathbb{Z}$.
\end{proposition}

In a similar way, the second term of the equivariant curvature, the moment
$\mu^{\Theta}$\ measures the variation of $\mathrm{hol}_{\phi}^{\Theta}%
(\gamma)$ with respect $\phi\in G$. Precisely, we have the following result
(see \cite[Proposition 8]{EquiHol})

\begin{proposition}
\label{InfHolonomy}For any $X\in\mathfrak{g}$ and $x\in M$ we have
$\mathrm{hol}_{\exp(X)}^{\Theta}(\tau^{x,X})=\mu_{X}^{\Theta}(x)$ where
$\tau^{x,X}(s)=\exp(sX)_{M}(x)$.
\end{proposition}

\subsection{Contractible base}

If $M$ is a contractible manifold, then several aspects can be simplified. As
in this paper we work with the spaces of connections and metrics, we study in
detail this case. If $M$ is contractible, then any principal $U(1)$-bundle is
trivializable, and hence it is enough to study the case of the trivial bundle
$\mathcal{U=}M\times U(1)\rightarrow M$.

As it is well known (see for example \cite{BCRS2}, \cite{RSW}), for the
trivial bundle $M\times U(1)\rightarrow M$ the action of $G$ on $M\times U(1)
$ is determined by a map $\alpha\colon G\times M\rightarrow\mathbb{R}%
/\mathbb{Z}$ characterized by the property
\[
\phi_{\mathcal{U}}(x,u)=(\phi_{M}(x),u\cdot\exp(2\pi i\cdot\alpha_{\phi}(x)).
\]
It satisfies the cocycle condition $\alpha_{\phi^{\prime}\cdot\phi}%
(x)=\alpha_{\phi}(x)+\alpha_{\phi^{\prime}}(\phi(x))$. Conversely any cocycle
determines an action of $G$ on $M\times U(1)$ by $U(1)$-bundle isomorphisms.
In this case the equivariant holonomy can be studied in terms of the cocycle
$\alpha_{\phi}(x)$ (e.g. see \cite{AnomaliesG}). For the trivial bundle,
a\ connection $\Theta$ is simply a form of the type$\footnote{For simplicity
in the notation, we use the same notation for forms on $M$ and $U(1)$ and its
pull-backs to $M\times U(1)$}$ $\Theta=\vartheta-2\pi i\lambda$ for a form
$\lambda\in\Omega^1(M)$ and where $\vartheta=z^-1dz$ is the Maurer-Cartan form
on $U(1)$.

\begin{proposition}
\label{holAlfa}If $\Theta=\vartheta-2\pi i\lambda$ is $G$-invariant, then for
any $\phi\in G$ and $\gamma\in\mathcal{C}_{x}^{\phi}(M)$ we have
\[
\mathrm{hol}_{\phi}^{\Theta}(\gamma)=%
%TCIMACRO{\tint \nolimits_{\gamma}}%
%BeginExpansion
{\textstyle\int\nolimits_{\gamma}}
%EndExpansion
\lambda-\alpha_{\phi}(x).
\]

\end{proposition}

\section{Equivariant differential characters}

In this section we define equivariant differential characters (of degree $2$)
as objects that satisfy properties similar to the equivariant log-holonomy. A
similar definition is introduced in \cite{LermanMalkin} (see Section
\ref{Other def} for details). Furthermore, a general definition of equivariant
differential cohomology for arbitrary order in the context of Deligne
Cohomology is introduced in \cite{kubel}. Although our definition is valid for
arbitrary manifolds, we study the case in which the manifold is contractible
because this is the case that we need in our applications to gauge theory and
the proofs are simpler because the equivariant $U(1)$-bundles can be studied
in terms of group cocycles.

First we define differential characters of degree $2$ as maps that satisfy the
same properties than the holonomy of a connection (see Section \ref{Other def}
for another equivalent definition)

\begin{definition}
\label{Dennoequi}A differential character of degree $2$ is a map $\chi
\colon\mathcal{C}(M)\rightarrow\mathbb{R}/\mathbb{Z}$ such that there exists a
closed $2$-form $\mathrm{curv}(\chi)\in\Omega^{2}(M)$ satisfying the following conditions

a) $\chi(\gamma^{\prime}\ast\gamma)=\chi(\gamma^{\prime})+\chi(\gamma)$ for
$\gamma,\gamma^{\prime}\in\mathcal{C}_{x}(M)$, $x\in M$.

b) If $\gamma\in\mathcal{C}(M)$ and $\gamma=\partial\sigma$ for $\sigma\in
C_{2}(M)$ then $\chi(\gamma)=\int_{\sigma}\mathrm{curv}(\chi
)\operatorname{mod}\mathbb{Z}$.
\end{definition}

The space of degree $2$ differential characters on $M$ is denoted by $\hat
{H}^{2}(M),$ and the map that assigns to a connection its holonomy induces a
bijection $\mathrm{Bund}_{U(1)}^{\nabla}(N)\simeq\hat{H}^{2}(N)$ (e.g. see
\cite[Theorem 2.5.1]{Kostant}), where $\mathrm{Bund}_{U(1)}^{\nabla}(N)$
denotes the set of principal $U(1)$-bundles with connection over a manifold
$N$ modulo isomorphisms (covering the identity on $M$).

In the equivariant case we can give a similar definition

\begin{definition}
\label{Defequi}A $G$-equivariant differential character is a map $\chi
\colon\mathcal{C}^{G}(M)\rightarrow\mathbb{R}/\mathbb{Z}$ such that there
exists a closed $G$-equivariant $2$-form $\mathrm{curv}_{G}(\chi
)=\mathrm{curv}(\chi)+\mu^{\chi}\in\Omega_{G}^{2}(M)$ satisfying the following conditions

i) $\chi(\phi^{\prime}\cdot\phi,\gamma\ast\gamma^{\prime})=\chi(\phi^{\prime
},\gamma^{\prime})+\chi(\phi,\gamma)$ for $\gamma\in\mathcal{C}^{\phi
}(M),\gamma^{\prime}\in\mathcal{C}_{\gamma(1)}^{\phi^{\prime}}(M)$.

ii) If $\zeta$ is a curve on $M$ such that $\zeta(0)=\gamma(0)$, and
$\gamma\in\mathcal{C}^{\phi}(M)$\ then we have $\chi(\phi,\overleftarrow
{\zeta}\ast\gamma\ast(\phi\cdot\zeta))=\chi(\phi,\gamma)$.

iii) If $\gamma\in\mathcal{C}^{e}(M)$ and $\gamma=\partial\sigma$ for a chain
$\sigma\in C_{2}(M)$ then $\chi(e,\gamma)=\int_{\sigma}\mathrm{curv}%
(\chi)\operatorname{mod}\mathbb{Z}$.

iv) For any $X\in\mathfrak{g}$ and $x\in M$ we have $\chi(\exp(X),\tau
^{x,X})=\mu_{X}^{\Theta}(x)$ where $\tau^{x,X}(s)=\exp(sX)_{M}(x)$.
\end{definition}

The space of $G$-equivariant differential characters on $M$ is denoted by
$\hat{H}_{G}^{2}(M)$. We have a natural map $\hat{H}_{G}^{2}(M)\rightarrow
\hat{H}^{2}(M)$. If $\Theta$ is a $G$-invariant connection on a $U(1)$-bundle,
we denote by $\mathrm{hol}_{G}^{\Theta}\in\hat{H}_{G}^{2}(M)$ the equivariant
differential character determined by $\mathrm{hol}_{G}^{\Theta}(\phi
,\gamma)=\mathrm{hol}_{\phi}^{\Theta}(\gamma)$ for $(\phi,\gamma
)\in\mathcal{C}^{G}(M)$.

\begin{example}
\label{beta}If $\beta\in\Omega^{1}(M)^{G}$ then we can define $\varsigma
(\beta)\in\hat{H}_{G}^{2}(M)$ by setting $\varsigma(\beta)(\phi,\gamma)=%
%TCIMACRO{\tint _{\gamma}}%
%BeginExpansion
{\textstyle\int_{\gamma}}
%EndExpansion
\beta\operatorname{mod}\mathbb{Z}$ for $\gamma\in\mathcal{C}^{\phi}(M)$. We
have $\mathrm{curv}_{G}(\varsigma(\beta))=D\beta$.
\end{example}

\begin{example}
\label{Exquotient}If $M/G$ is a manifold, $\pi\colon M\rightarrow M/G$ is the
projection and $\chi\in\hat{H}^{2}(M/G)$, then for any $\gamma\in
\mathcal{C}^{\phi}(M)$ the curve $\pi\circ\gamma$ is a closed loop on $M/G$.
We define $(\pi_{G}^{\ast}\chi)(\phi,\gamma)=\chi(\pi\circ\gamma)$ and we have
$\pi_{G}^{\ast}\chi\in\hat{H}_{G}^{2}(M)$ and $\mathrm{curv}_{G}(\pi_{G}%
^{\ast}\chi)=\pi^{\ast}(\mathrm{curv}(\chi))$.
\end{example}

\begin{example}
If $\xi\in\mathrm{Hom}(G,\mathbb{R}/\mathbb{Z})$ we define $\chi(\phi
,\gamma)=\xi(\phi)$ for $\phi\in G$ and $\gamma\in\mathcal{C}^{\phi}(M)$. We
have $\chi\in\hat{H}_{G}^{2}(M)$ and $\mathrm{curv}_{G}(\chi)=d\xi$, where
$d\xi\in\mathrm{Hom}(\mathfrak{g},\mathbb{R})\subset\mathrm{Hom}%
(\mathfrak{g},\Omega^{0}(M))$ is the differential of $\xi$.
\end{example}

\begin{remark}
\label{reparam}It follows from the conditions i) and iii) that if $\gamma$ and
$\gamma^{\prime}$ differ by a reparametrization, then $\chi(\phi,\gamma
)=\chi(\phi,\gamma^{\prime})$.
\end{remark}

The condition iv) is equivalent to a weaker condition that it is easier to
check in practice.

\begin{proposition}
If the conditions i) and iii) are satisfied, then the condition iv) is
equivalent to the following condition

iv') For any $X\in\mathfrak{g}$ and $x\in M$ we have $\left.  \frac{d}%
{dt}\right\vert _{t=0}\chi(\exp(tX),\nu_{t}^{x,X})=\mu_{X}^{\Theta}(x)$ where
$\nu_{t}^{x,X}(s)=\exp(stX)_{M}(x)$.
\end{proposition}

\begin{proof}
Clearly iv') follows from iv). We prove the converse. We define $k(t)=\chi
(\exp(tX),\nu_{t}^{x,X})$ and we have $k(0)=0$ and $k(1)=\chi(\exp
(X),\tau^{x,X})$. As the curves $\nu_{t+s}^{x,X}$ and $\nu_{t}^{x,X}\ast
\nu_{s}^{x,X}$ differ by a reparametrization, by the condition i) and Remark
\ref{reparam} we have $k(t+t^{\prime})=k(t)+k(t^{\prime})$. Taking the
derivative we obtain $\frac{dk}{dt}(t)=\underset{h\rightarrow0}{\lim}%
\frac{k(t+h)-k(t)}{h}=\underset{h\rightarrow0}{\lim}\frac{k(h)}{h}=\mu
_{X}^{\Theta}(x)$, and by integration it follows that $k(t)=t\mu_{X}^{\Theta
}(x)$ for any $t$. By taking $t=1$ we obtain the condition iv).
\end{proof}

\begin{lemma}
\label{PropDC}If $M$ is connected, $\chi\in\hat{H}_{G}^{2}(M)$ and $\phi
$,$\phi^{\prime}\in G$, $\gamma,\gamma^{\prime}\in\mathcal{C}^{\phi}(M)$ then
we have

a) $\chi(\phi^{-1},\overleftarrow{\gamma})=-\chi(\phi,\gamma)$.

b)$\ \chi(\phi^{\prime}\cdot\phi\cdot(\phi^{\prime})^{-1},\phi^{\prime}%
\cdot\gamma)=\chi(\phi,\gamma)$.
\end{lemma}

\begin{proof}
a) By Properties i) and iii) we have $\chi(\phi^{-1},\overleftarrow{\gamma
})+\chi(\phi,\gamma)=\chi(e,\overleftarrow{\gamma}\ast\gamma)=0$.

b) If $\upsilon$ is a curve on $M$ joining $\gamma(0)$ and $\phi_{M}^{\prime
}(\gamma(0))$ then by conditions i), ii) and property a) we have
\[
\chi(\phi,\upsilon)=\chi(\phi,\overleftarrow{\gamma}\ast\upsilon\ast
(\phi^{\prime}\cdot\gamma))=-\chi(\phi,\gamma)+\chi(\phi,\upsilon)+\chi
(\phi^{\prime}\cdot\phi\cdot(\phi^{\prime})^{-1},\phi^{\prime}\cdot\gamma),
\]
and hence $\chi(\phi^{\prime}\cdot\phi\cdot(\phi^{\prime})^{-1},\phi^{\prime
}\cdot\gamma)=\chi(\phi,\gamma)$.
\end{proof}

The next construction will appear frequently in our applications to Gauge
theory. Let the groups $G$ and $H$ act on the manifolds $M$ and $N$
respectively, let $\rho\colon H\rightarrow G$ be a Lie group homomorphisms and
let $f\colon N\rightarrow M$ be a $\rho$-equivariant map.

\begin{proposition}
\label{pullback}If $f\colon N\rightarrow M$ is $\rho$-equivariant, then any
differential character $\chi\in\hat{H}_{G}^{2}(M)$ defines a $H$-equivariant
differential character $(f,\rho)^{\ast}\chi\in\hat{H}_{H}^{2}(N)$ by
$((f$,$\rho)^{\ast}\chi)(\phi,\gamma)=\chi(\rho(\phi),f\circ\gamma)$. The
$H$-equivariant curvature of $(f,\rho)^{\ast}\chi$ is $(f$,$\rho)^{\ast
}(\mathrm{curv}_{G}(\chi))$.
\end{proposition}

It is shown in Section \ref{Other def} that $\hat{H}_{G}^{2}(M)$ is isomorphic
to the space $G$-equivariant $U(1)$-bundles with connection over $M$ modulo
isomorphisms (covering the identity on $M$). Next we show that in the
particular case in which $M$ is contractible, it is possible to determine a
concrete bundle with connection that corresponds to a $G$-equivariant
differential character $\chi\in\hat{H}_{G}^{2}(M)$. In Section
\ref{SectEquiCS} we apply this construction in order to define the
Chern-Simons bundles.

\begin{theorem}
\label{PropAction}Let $M$ be a contractible manifold, $\chi\in\hat{H}_{G}%
^{2}(M)$ and let $\lambda\in\Omega^{1}(M)$ be a $1$-form such that
$d\lambda=\mathrm{curv}(\chi)$. Then there exists a unique lift of the action
of $G$ to $M\times U(1)$ by $U(1)$-bundle\ automorphisms such that
$\Theta=\vartheta-2\pi i\lambda$ is $G$-invariant and $\chi=\mathrm{hol}%
_{G}^{\Theta}$. Precisely, the action is defined by the cocycle $\alpha_{\phi
}(x)=\int_{\gamma}\lambda-\chi(\phi,\gamma)$ for any $\gamma\in\mathcal{C}%
_{x}^{\phi}(M)$.
\end{theorem}

\begin{proof}
First we show that $\alpha_{\phi}(x)=\int_{\gamma}\lambda-\chi(\phi,\gamma)$
does not depend on the curve $\gamma\in\mathcal{C}_{x}^{\phi}(M)$. If
$\gamma,\gamma^{\prime}\in\mathcal{C}_{x}^{\phi}(M)$ then $\overleftarrow
{\gamma}\ast\gamma^{\prime}$ is a closed loop on $M$. If $\Sigma$ is a
submanifold of dimension $2$ such that $\partial\Sigma=\overleftarrow{\gamma
}\ast\gamma^{\prime}$ (it exists by the contractibility of $M$) then by Lemma
\ref{PropDC} a)$\ $we have%
\[
\chi(\phi,\gamma^{\prime})-\chi(\phi,\gamma)=\chi(\phi,\overleftarrow{\gamma
}\ast\gamma^{\prime})=%
%TCIMACRO{\tint \nolimits_{\Sigma}}%
%BeginExpansion
{\textstyle\int\nolimits_{\Sigma}}
%EndExpansion
d\lambda=%
%TCIMACRO{\tint \nolimits_{\overleftarrow{\gamma}\ast\gamma^{\prime}}}%
%BeginExpansion
{\textstyle\int\nolimits_{\overleftarrow{\gamma}\ast\gamma^{\prime}}}
%EndExpansion
\lambda=%
%TCIMACRO{\tint \nolimits_{\gamma^{\prime}}}%
%BeginExpansion
{\textstyle\int\nolimits_{\gamma^{\prime}}}
%EndExpansion
\lambda-%
%TCIMACRO{\tint \nolimits_{\gamma}}%
%BeginExpansion
{\textstyle\int\nolimits_{\gamma}}
%EndExpansion
\lambda,
\]
and hence $\int_{\gamma}\lambda-\chi(\phi,\gamma)=\int_{\gamma^{\prime}%
}\lambda-\chi(\phi,\gamma^{\prime})$.

Next we prove that $\alpha$ satisfies the cocycle condition. If $\gamma
\in\mathcal{C}_{x}^{\phi}(M)$ and $\gamma^{\prime}\in\mathcal{C}_{x}%
^{\phi^{\prime}}(M)$ we have%
\begin{align*}
\alpha_{\phi^{\prime}\cdot\phi}(x) &  =%
%TCIMACRO{\tint \nolimits_{(\phi\cdot\gamma^{\prime})\ast\gamma}}%
%BeginExpansion
{\textstyle\int\nolimits_{(\phi\cdot\gamma^{\prime})\ast\gamma}}
%EndExpansion
\lambda-\chi(\phi^{\prime}\cdot\phi,(\phi\cdot\gamma^{\prime})\ast\gamma)\\
&  =%
%TCIMACRO{\tint \nolimits_{\phi\cdot\gamma^{\prime}}}%
%BeginExpansion
{\textstyle\int\nolimits_{\phi\cdot\gamma^{\prime}}}
%EndExpansion
\lambda+%
%TCIMACRO{\tint \nolimits_{\gamma}}%
%BeginExpansion
{\textstyle\int\nolimits_{\gamma}}
%EndExpansion
\lambda-\chi(\phi^{\prime}\cdot\phi\cdot(\phi^{\prime})^{-1},\phi\cdot
\gamma^{\prime})-\chi(\phi,\gamma)\\
&  =\alpha_{\phi}(x)+\alpha_{\phi^{\prime}}(\phi(x))
\end{align*}

If $x,x^{\prime}\in M$ and $\zeta$ is a curve on $M$ with $\zeta(0)=x$,
$\zeta(1)=x^{\prime}$ and $\gamma\in\mathcal{C}_{x}^{\phi}(M)$\ then
$\overleftarrow{\zeta}\ast\gamma\ast(\phi\cdot\zeta)\in\mathcal{C}_{x^{\prime
}}^{\phi}(M)$ and by property ii) we have
\begin{align*}
\alpha_{\phi}(x^{\prime}) &  =%
%TCIMACRO{\tint \nolimits_{\overleftarrow{\zeta}\ast\gamma\ast(\phi\cdot\zeta
%)}}%
%BeginExpansion
{\textstyle\int\nolimits_{\overleftarrow{\zeta}\ast\gamma\ast(\phi\cdot\zeta
)}}
%EndExpansion
\lambda-\chi(\phi,\overleftarrow{\zeta}\ast\gamma\ast(\phi\cdot\zeta))=%
%TCIMACRO{\tint \nolimits_{\phi\cdot\zeta}}%
%BeginExpansion
{\textstyle\int\nolimits_{\phi\cdot\zeta}}
%EndExpansion
\lambda+%
%TCIMACRO{\tint \nolimits_{\gamma}}%
%BeginExpansion
{\textstyle\int\nolimits_{\gamma}}
%EndExpansion
\lambda-%
%TCIMACRO{\tint \nolimits_{\zeta}}%
%BeginExpansion
{\textstyle\int\nolimits_{\zeta}}
%EndExpansion
\lambda-\chi(\phi,\gamma)\\
&  =%
%TCIMACRO{\tint \nolimits_{\zeta}}%
%BeginExpansion
{\textstyle\int\nolimits_{\zeta}}
%EndExpansion
\phi_{M}^{\ast}\lambda+%
%TCIMACRO{\tint \nolimits_{\gamma}}%
%BeginExpansion
{\textstyle\int\nolimits_{\gamma}}
%EndExpansion
\lambda-%
%TCIMACRO{\tint \nolimits_{\zeta}}%
%BeginExpansion
{\textstyle\int\nolimits_{\zeta}}
%EndExpansion
\lambda-\chi(\phi,\gamma)=%
%TCIMACRO{\tint \nolimits_{\zeta}}%
%BeginExpansion
{\textstyle\int\nolimits_{\zeta}}
%EndExpansion
(\phi_{M}^{\ast}\lambda-\lambda)+\alpha_{\phi}(x).
\end{align*}

It follows from this condition that $\alpha_{\phi}(x)$ is differentiable with
respect $x$ and that
\begin{equation}
d\alpha_{\phi}=\phi_{M}^{\ast}\lambda-\lambda.\label{dalfa}%
\end{equation}

The differentiability of $\alpha$ with respect to $\phi$ follows from
condition iv) in the definition of equivariant differential character.

We define the connection form $\Theta=\vartheta-2\pi i\lambda\in\Omega
^{1}(M\times U(1),i\mathbb{R})$. For every $\phi\in\mathcal{G}$, using
equation (\ref{dalfa}) we obtain
\begin{align*}
\phi_{M\times U(1)}^{\ast}\Theta & =\phi_{M\times U(1)}^{\ast}\vartheta-2\pi
i\phi_{M}^{\ast}\lambda=(\vartheta+2\pi id\alpha_{\phi})-2\pi i\phi_{M}^{\ast
}\lambda\\
& =(\vartheta+2\pi i(\phi_{M}^{\ast}\lambda-\lambda))-2\pi i\phi_{M}^{\ast
}\lambda=\Theta.
\end{align*}

Hence $\Theta$ is $G$-invariant and from Proposition \ref{holAlfa} it
follows\ that $\mathrm{hol}_{G}^{\Theta}=\chi$.
\end{proof}

From the preceding theorem we conclude the following

\begin{corollary}
\label{ExistenceBundle}If $M$ is contractible, then for any $G$-equivariant
differential character $\chi\in\hat{H}_{G}^{2}(M)$ there exists a
$G$-equivariant $U(1)$-bundle with connection $(\mathcal{U},\Theta)$ such that
$\chi(\phi,\gamma)=\mathrm{hol}_{\phi}^{\Theta}(\gamma)$.
\end{corollary}

If we denote by $\mathrm{Bund}_{U(1)}^{\nabla,G}(M)$\ the space of
$G$-equivariant $U(1)$-bundles with connection over $M$ modulo isomorphisms
(covering the identity on $M$), then from Corollary \ref{ExistenceBundle} and
Theorem \ref{anomaly} a) we obtain

\begin{corollary}
\label{Cor iso}If $M$ is contractible, then the map that assigns to a
$G$-equivariant $U(1)$-bundle with connection its $G$-equivariant log-holonomy
determines a bijection $\mathrm{Bund}_{U(1)}^{\nabla,G}(M)\simeq\hat{H}%
_{G}^{2}(M).$
\end{corollary}

This result is generalized to arbitrary manifolds in Section \ref{Other def}.

\begin{remark}
We conclude that a $G$-equivariant differential character determines a
$G$-equivariant $U(1)$-bundle with connection modulo an isomorphism. However,
in order to determine a concrete bundle with connection, it is necessary to
give additional information. In the case of contractible manifolds it is
enough to give a form $\lambda\in\Omega^{1}(M)$ such that $d\lambda
=\mathrm{curv}(\chi)$.
\end{remark}

Theorem \ref{anomaly} b) can be reinterpreted in terms of differential
characters. Precisely we have the following

\begin{proposition}
\label{anomalyC}If $M$ is contractible then the $G$-equivariant $U(1)$-bundle
determined (up to an isomorphism) by $\chi\in\hat{H}_{G}^{2}(M)$ is trivial if
and only if there exists $\beta\in\Omega^{1}(M)^{G}$ such that $\chi
=\varsigma(\beta)$.
\end{proposition}

\subsection{Projectable Differential Characters\label{projectable}}

Suppose that $M/G$ is a manifold, the projection $\pi\colon M\rightarrow M/G $
is smooth. We say that a differential character $\chi\in\hat{H}^{2}(M)$ is
$\pi$-projectable if there exists $\underline{\chi}\in\hat{H}^{2}(M/G)$ such
that $\chi=\pi_{G}^{\ast}\underline{\chi}$. A necessary condition for $\chi$
to be $\pi$-projectable is $\mu^{\chi}=0$. For free actions, this condition is
also sufficient

\begin{proposition}
\label{quot}If $G$ acts freely on a contractible manifold $M\ $and $\pi\colon
M\rightarrow M/G$ is a principal $G$-bundle, then $\chi\in\hat{H}_{G}^{2}(M)$
is $\pi$-projectable if and only if $\mu^{\chi}=0$.
\end{proposition}

\begin{proof}
Let $\lambda\in\Omega^{1}(M)$ be a $1$-form such that $d\lambda=\mathrm{curv}%
(\chi)$ and let $\Theta=\vartheta-2\pi i\lambda$ be the corresponding
connection. For any $X\in\mathfrak{g}$ we have $\iota_{X_{M\times U(1)}}%
\Theta=2\pi i\mu_{X}^{\chi}=0$, and as $\Theta$ is $G$-invariant, it projects
onto a connection $\underline{\Theta}\ $on $(M\times U(1))/G\rightarrow M/G$.
It is easily seen that $\chi=\pi_{G}^{\ast}\mathrm{hol}_{G}^{\underline
{\Theta}}$.
\end{proof}

We need also the following generalization of the preceding result, that can be
proved in a similar way

\begin{proposition}
\label{quot2}Let $H\subset G$ be a Lie subgroup of $G$ that acts freely on a
contractible manifold $M\ $and such that $\pi\colon M\rightarrow M/H$ is a
principal $H$-bundle. Then $G$ acts on $M/H$, and if $\phi\in G$ and
$\gamma\in\mathcal{C}^{\phi}(M)$ then $\pi\circ\gamma\in\mathcal{C}^{\phi
}(M/H)$. If $\chi\in\hat{H}_{G}^{2}(M)$ then there exists $\underline{\chi}%
\in\hat{H}_{G}^{2}(M/H)$\ such that $\chi(\phi,\gamma)=\underline{\chi}%
(\pi\circ\gamma)$\ if and only if $\mu^{\chi}|_{\mathfrak{h}}=0$.
\end{proposition}

\begin{remark}
\label{quot20}If $H\subset G$ is normal closed subgroup, we can also consider
the action of the quotient group $G/H$ on $M/H$ and if $\mu^{\chi
}|_{\mathfrak{h}}=0$ we obtain an element of $\hat{H}_{G/H}^{2}(M/H)$.
\end{remark}

\subsection{Other definitions and arbitrary manifolds}\label{Other def}%
\footnote{The results of this section are not necessary for the study of the
Chern-Simons bundles.}

The Definition \ref{Dennoequi} of differential character is a direct
generalization of the concept of holonomy, but it is not the usual definition.
Usually the differential characters are defined on $Z_{1}(M)$ and not on
$\mathcal{C}(M)$ (e.g. see Section \ref{CCS}). In this Section we show that
both definitions are equivalent, and we study their generalization\ to the
equivariant case. First we need the following technical result

\begin{definition}
A chain $\upsilon\in C_{2}(M)$ is a thin chain if $\int_{\upsilon}\alpha=0 $
for any $\alpha\in\Omega^{2}(M)$.
\end{definition}

\begin{lemma}
\label{Aditivo} If $\gamma_{1},\gamma_{2}\ $are curves on $M$ with $\gamma
_{2}(0)=\gamma_{1}(1)$, then $\gamma_{1}\ast\gamma_{2}=\gamma_{1}+\gamma
_{2}+\partial\upsilon$, with $\upsilon\in C_{2}(M)$ a thin chain.
\end{lemma}

The explicit form of the chain $\upsilon$ can be found for example in
\cite[page 64]{Greenberg}.

If $\chi$ is a differential character as defined in Definition \ref{Dennoequi}%
, then we can extend $\chi$ to $Z_{1}(M)$ in the following way: if $z\in
Z_{1}(M)$ then we have $z=\gamma+\partial\sigma$, with $\gamma\in
\mathcal{C}(M)$ and $\sigma\in C_{2}(M)$\footnote{This fact follows from the
surjectivity of the map $\pi_{1}(M)\rightarrow H_{1}(M)$.}. Then we define
$\chi(z)=\chi(\gamma)+\int_{\sigma}\mathrm{curv}(\chi)\operatorname{mod}%
\mathbb{Z}$. That $\chi$ is well defined and that it is a group homomorphism
$\chi\colon Z_{1}(M)\rightarrow\mathbb{R}/\mathbb{Z}$ easily follows from
Lemma \ref{Aditivo}. We conclude that a definition equivalent to Definition
\ref{Dennoequi} is the following

\begin{definition}
A differential character of degree $2$ is a group homomorphism $\chi\colon
Z_{1}(M)\rightarrow\mathbb{R}/\mathbb{Z}$ such that there exists a closed $2
$-form $\mathrm{curv}(\chi)\in\Omega^{2}(M)$ satisfying $\chi(\partial
\sigma)=\int_{\sigma}\mathrm{curv}(\chi)\operatorname{mod}\mathbb{Z}$ for any
$\sigma\in C_{2}(M)$.
\end{definition}

In \cite{LermanMalkin} a definition of $G$-equivariant differential character
similar to the preceding one is introduced. Let $C_{1,G}(M)$ be the free
abelian group generated by pairs $(\phi,\gamma),$with $\gamma$ a curve on $M$
and $\phi\in G$. We define $\partial_{G}\colon C_{1,G}(M)\rightarrow C_{0}(M)$
by setting $\partial_{G}(\phi,\gamma)=\phi(\gamma(1))-\gamma(0)$ and
$Z_{1,G}(M)=\ker\partial_{G}$. Note that $Z_{1,G}(M)$ is generated by chains
of the form $(\phi_{1},\gamma_{1})+\ldots+(\phi_{n},\gamma_{n})$ that satisfy
the following condition%
\begin{equation}
\phi_{i}(\gamma_{i}(1))=\gamma_{i+1}(0)\text{, }i=1,\ldots,n-1\text{ and }%
\phi_{n}(\gamma_{n}(1))=\gamma_{1}(0)\label{condition}%
\end{equation}
In particular, if $\gamma\in\mathcal{C}^{\phi}(M)$ then $(\phi^{-1},\gamma)\in
Z_{1,G}(M)$.

\begin{definition}
\label{LermanMalkin def}A Lerman-Malkin $G$-equivariant differential character
is a group homomorphism $\eta\colon Z_{1,G}(M)\rightarrow\mathbb{R}%
/\mathbb{Z}$ such that there exists a closed $G$-equivariant $2$-form
$\mathrm{curv}_{G}(\eta)=\mathrm{curv}(\eta)+\mu^{\eta}\in\Omega_{G}^{2}(M)$
satisfying the following conditions

a) $\eta((\phi,\xi)+(\phi^{-1},\phi\cdot\overleftarrow{\xi}))=0$ for any curve
$\xi$ on $M$ and $\phi\in G$.

b) If $\sigma\in C_{2}(M)$ then $\eta(e,\partial\sigma)=\int_{\sigma
}\mathrm{curv}(\eta)\operatorname{mod}\mathbb{Z}$.

c) For any $X\in\mathfrak{g}$ and $x\in M$ we have $\eta(\exp(-X),\tau
^{x,X})=\mu_{X}^{\Theta}(x)$ where $\tau^{x,X}(s)=\exp(sX)_{M}(x)$

d) For any $x\in M$ the map $\phi\mapsto\eta(\phi,c_{x})$ determines a group
homomorphism $G_{x}\rightarrow\mathbb{R}/\mathbb{Z}$, where\ $G_{x}$ is the
isotropy group of $x$, and $c_{x}$ the constant curve with value $x$.
\end{definition}

Next we show that if $M$ is connected, then Definition \ref{LermanMalkin def}
is equivalent to Definition \ref{Defequi}. Let $\chi$ be an equivariant
differential character and let $z=%
%TCIMACRO{\tsum _{i=1}^{n}}%
%BeginExpansion
{\textstyle\sum_{i=1}^{n}}
%EndExpansion
(\phi_{i},\gamma_{i})$ be a cycle that satisfies the condition
(\ref{condition}). We chose curves $\tau_{i}\in\mathcal{C}_{\gamma_{i}%
(1)}^{\phi_{i}}(M)$ on $M$ and we define $\eta(z)=%
%TCIMACRO{\tsum _{i=1}^{n}}%
%BeginExpansion
{\textstyle\sum_{i=1}^{n}}
%EndExpansion
\chi(e,\gamma_{1}\ast\tau_{1}\ast\ldots\ast\gamma_{n}\ast\tau_{n})-%
%TCIMACRO{\tsum _{i=1}^{n}}%
%BeginExpansion
{\textstyle\sum_{i=1}^{n}}
%EndExpansion
\chi(\phi_{i},\tau_{i})$. It is easily seen that this definition is
independent of the curves $\tau_{i}$ chosen.

Clearly $\eta$ determines a group homomorphism $\eta\colon Z_{1,G}%
(M)\rightarrow\mathbb{R}/\mathbb{Z}$\ that satisfies the conditions b), c) and
d) in Definition \ref{LermanMalkin def}.\ Next we prove that condition\ a)\ is
basically equivalent to condition ii) in Definition \ref{Defequi}. If $\xi$ is
a curve on $M$ and $\phi\in G$, we choose $\tau_{1}\in\mathcal{C}_{\xi
(1)}^{\phi}(M)$ and $\tau_{2}\in\mathcal{C}_{\phi(\xi(0))}^{\phi^{-1}}(M)$.
Then by the definition of $\eta$ and the conditions i) and ii) in Definition
\ref{Defequi} we have%
\begin{multline*}
\eta((\phi,\xi)+(\phi^{-1},\phi\cdot\overleftarrow{\xi}))=\chi(e,\xi\ast
\tau_{1}\ast(\phi\cdot\overleftarrow{\xi})\ast\tau_{2})-\chi(\phi,\tau
_{1})-\chi(\phi^{-1},\tau_{2})\\
=\chi(\phi,\xi\ast\tau_{1}\ast(\phi\cdot\overleftarrow{\xi}))+\chi(\phi
^{-1},\tau_{2})-\chi(\phi,\tau_{1})-\chi(\phi^{-1},\tau_{2})=0\text{.}%
\end{multline*}

Conversely, If $\eta$ satisfies the conditions of Definition
\ref{LermanMalkin def}, then we can define $\chi\colon\mathcal{C}%
^{G}(M)\rightarrow\mathbb{R}/\mathbb{Z}$ by setting $\chi(\phi,\gamma
)=\eta(\phi^{-1},\gamma)$ for $\gamma\in\mathcal{C}^{\phi}(M)$. Clearly
conditions iii) and iv) are satisfied. We prove condition ii). If $\phi\in G$,
we consider the mapping torus $\pi\colon M\times I\rightarrow M_{\phi}$. For
any curve $\gamma$ on $M$, we define $\underline{\gamma}=\pi\circ(\gamma
\times\mathrm{id}_{I})$. The Lerman-Malkin $G$-equivariant differential
character $\eta$ projects onto a differential character $\underline{\eta}%
\in\hat{H}^{2}(M_{\phi})$ such that $\underline{\eta}(\underline{\gamma}%
)=\eta(\phi^{-1},\gamma)$ for any $\gamma\in\mathcal{C}^{\phi}(M)$. Then we
have%
\begin{multline*}
\chi(\phi,\overleftarrow{\zeta}\ast\gamma\ast(\phi\cdot\zeta))=\eta(\phi
^{-1},\overleftarrow{\zeta}\ast\gamma\ast(\phi\cdot\zeta))=\underline{\eta
}(\underline{\overleftarrow{\zeta}\ast\gamma\ast(\phi\cdot\zeta)})\\
=\underline{\eta}(\underline{\overleftarrow{\zeta}}\ast\underline{\gamma}%
\ast\underline{(\phi\cdot\zeta)})=\underline{\eta}(\underline{\overleftarrow
{\zeta}}\ast\underline{\gamma}\ast\underline{\zeta})=\underline{\eta
}(\underline{\gamma})=\eta(\phi^{-1},\gamma)=\chi(\phi,\gamma).
\end{multline*}

The condition i) can be proved in a similar way, by replacing $M_{\phi}$ with
$(M\times E)/H$, where $H$ is the subgroup generated by $\phi$ and
$\phi^{\prime}$, $E$ is a manifold in which $H$ acts freely, and
$\mathrm{id}_{I}$ is replaced by two curves $\tau\in\mathcal{C}^{\phi}(E)$,
$\tau^{\prime}\in\mathcal{C}^{\phi^{\prime}}(E)$.

It is proved in \cite[Theorem 4.3.1]{LermanMalkin} that there exists a
bijection between the space of Lerman-Malkin $G$-equivariant differential
characters and $\mathrm{Bund}_{U(1)}^{\nabla,G}(M)$. This result generalizes
Corollary \ref{Cor iso} to the case of an arbitrary manifold $M$.

\section{Integrated equivariant Cheeger-Chern-Simons differential characters
\label{EDC}}

In this section we apply the preceding constructions to the case of the space
of connections $\mathcal{A}_{P}$ on a principal bundle and the action of the
group of automorphisms. We show that the Cheeger-Chern-Simons construction
determines in a natural way an equivariant differential character $\Xi
_{P}^{\vec{p}}$ on $\mathcal{A}_{P}$. In Section \ref{SectEquiHolCS} we show
that in the case of a trivial bundle $\Xi_{P}^{\vec{p}}$ coincides (up to a
sign) with the equivariant holonomy of the Chern-Simons line bundle. For
arbitrary bundles the Chern-Simons bundle can be defined by applying Theorem
\ref{PropAction}\ to $\Xi_{P}^{\vec{p}}$. We also show that this bundle is
isomorphic to the bundle defined in \cite{CSconnections}. In the later
sections we study how other constructions in Chern-Simons theory can be
derived form the equivariant differential character $\Xi_{P}^{\vec{p}}$.

\subsection{Cheeger-Chern-Simons differential characters and the Chern-Simons
action\label{CCS}}

We recall the properties of the Cheeger-Chern-Simons\ differential characters
introduced in \cite{Cheeger}. If $M$ is an oriented manifold then a
differential character of degree $k$\ is a homomorphism $\chi\colon
Z_{k-1}(M)\rightarrow\mathbb{R}/\mathbb{Z}$ such that there exists $\omega
\in\Omega^{k}(M)$ (called the curvature of $\chi)$ satisfying $\chi(\partial
u)=\int_{u}\omega\operatorname{mod}\mathbb{Z}$ for any cycle $u\in Z_{k}(M)$.

Let $G$ be a Lie group with a finite number of connected components. A
characteristic pair of degree $r$ for the group $G$ is a pair $\vec
{p}=(p,\Upsilon)$, where $p\in I^{r}(G)$ is a Weil polynomial of degree $r$,
$\Upsilon\in H^{2r}(\mathbf{B}G,\mathbb{Z})$ a characteristic class, and they
are compatible in the sense that they determine the same real characteristic
class on $H^{2r}(\mathbf{B}G,\mathbb{R})$. We denote by $I_{\mathbb{Z}}%
^{r}(G)$ the subset of elements $p\in I^{r}(G)$ that are compatible with a
characteristic class $\Upsilon\in H^{2r}(\mathbf{B}G,\mathbb{Z})$.

For any principal $G$-bundle $P\rightarrow M$ with connection $A$, the pair
$\vec{p}$ determines in a natural way a differential character $\xi_{A}%
^{\vec{p}}\in\hat{H}^{2r}(M)$ with curvature $p(F_{A})\in\Omega^{2r}(M)$. In
particular, for any $2r$-dimensional chain $u\in C_{2r}(M)$ we have $\xi
_{A}^{\vec{p}}(\partial u)=\int_{u}p(F_{A})\operatorname{mod}\mathbb{Z}$. We
recall that natural means that for any principal $G$-bundle $P^{\prime
}\rightarrow N^{\prime}$ and any $G$-bundle map $F\colon P^{\prime}\rightarrow
P$ we have
\begin{equation}
\chi_{F^{\ast}A}=f^{\ast}(\chi_{A}),\label{naturality}%
\end{equation}
where $f\colon N^{\prime}\rightarrow N$ is the map induced by $F$. If
$A^{\prime}$ is another connection on $P$, then for any $u\in Z_{2r-1}(M)$ we
have
\begin{equation}
\chi_{A^{\prime}}(u)=\chi_{A}(u)+%
%TCIMACRO{\tint _{u}}%
%BeginExpansion
{\textstyle\int_{u}}
%EndExpansion
Tp(A^{\prime},A),\label{AAprima}%
\end{equation}
where $Tp(A,A^{\prime})=r\int_{0}^{1}p(a,F_{t},\overset{(r-1)}{\ldots}%
,F_{t})dt$ is the Chern-Simons transgression form, with $a=A-A^{\prime}%
\in\Omega^{1}(M,\mathrm{ad}P)$ and $F_{t}$ the curvature of the connection
$A_{t}=tA+(1-t)A^{\prime}$. Furthermore, we have the following result (see
e.g. \cite[Proposition 2.9]{Cheeger})

\begin{lemma}
\label{LieCS}If $A_{t}$ is a smooth $1$-parametric family of connections on
$P$ with $\dot{A}_{0}=a\in\Omega^{1}(M,\mathrm{ad}P)$, then $\left.  \frac
{d}{dt}\right\vert _{t=0}\chi_{A_{t}}(u)=r\int_{u}p(a,F_{0},\overset
{(r-1)}{\ldots},F_{0})$ for every $u\in Z_{2r-1}(M)$.
\end{lemma}

If $M$ is compact and without boundary of dimension $\dim M=2r-1$ and
$P\rightarrow M\,\ $is a principal $G$-bundle, then the Chern-Simons action
$\mathrm{CS}^{\vec{p}}\colon\mathcal{A}_{P}\rightarrow\mathbb{R}/\mathbb{Z}$
is defined (e.g. see \cite{DW})\ by setting $\mathrm{CS}^{\vec{p}}(A)=\xi
_{A}^{\vec{p}}(M)$. It follows from equation (\ref{AAprima}) that if
$A,A^{\prime}$ are two connections on $P$, then $\mathrm{CS}^{\vec{p}%
}(A)=\mathrm{CS}^{\vec{p}}(A^{\prime})+\int_{M}Tp(A,A^{\prime})$. Moreover,
form the naturality condition (\ref{naturality}) we conclude that if $(P,A)$
is isomorphic to $(P^{\prime},A^{\prime})$ then $\mathrm{CS}^{\vec{p}%
}(A)=\mathrm{CS}^{\vec{p}}(A^{\prime})$. It also follows from
(\ref{naturality}) that if $P$ is a\ trivializable bundle, and $A_{0}$ the
connection associated to a trivialization then $\mathrm{CS}^{\vec{p}}%
(A_{0})=0$,\ and hence $\mathrm{CS}^{\vec{p}}(A)=\int_{M}Tp(A,A_{0})$. In the
particular case in which $r=2$, $G=SU(2)$, $P$ is a trivializable bundle with
a section $S\colon M\rightarrow P$ and\emph{\ }$p(X)=\frac{1}{8\pi^{2}%
}\mathrm{tr}(X^{2})$, then\emph{\ }we have$\ \mathrm{CS}^{\vec{p}}(A)=\frac
{1}{8\pi^{2}}\int_{M}\mathrm{tr}(\alpha\wedge d\alpha+\frac{2}{3}\alpha
\wedge\alpha\wedge\alpha)$, where $\alpha=S^{\ast}A$. Hence in this case
$\mathrm{CS}^{\vec{p}}$ which coincides with the classical Chern-Simons action
(e.g. see \cite{Freed1}).

\begin{remark}
For trivializable bundles the Chern-Simons action $\mathrm{CS}^{\vec{p}}$ only
depends on the polynomial $p\in I_{\mathbb{Z}}^{r}(G)$ and it is independent
of the characteristic class $\Upsilon$. In this case we denote the
Chern-Simons action by $\mathrm{CS}^{p}$.
\end{remark}

\subsection{Geometry of the space of connections\label{SectGeoConn}}

Let $P\rightarrow M$ be a principal $G$-bundle, and let $\mathcal{A}_{P}$
be\ the space of principal connections on this bundle, considered as an
infinite dimensional Fr\'{e}chet manifold. As $\mathcal{A}_{P}$ is an affine
space modeled on $\Omega^{1}(M,\mathrm{ad}P)$, we have canonical isomorphisms
$T_{A}\mathcal{A}_{P}\simeq\Omega^{1}(M,\mathrm{ad}P)$ for any $A\in
\mathcal{A}_{P}$. We denote by $\mathrm{Aut}P$ the group of $G$-bundle
automorphism of $P$, and by $\mathrm{Gau}P$ the subgroup of bundle
automorphism covering the identity on $M$. The Lie algebra of $\mathrm{Aut}P$
is the space of $G$-invariant vector fields on $P$, $\mathrm{aut}%
P\subset\mathfrak{X}(P)$, and the Lie algebra of $\mathrm{Gau}P$ is the
subspace $\mathrm{gau}P$ of vertical $G$-invariant vector fields (see
\cite{Abati}, \cite{Wockel} for more details on the Lie group structure of
$\mathrm{Aut}P$).\ The group $\mathrm{Aut}P$ acts in a natural way on
$\mathcal{A}_{P}$. If $M$ is oriented, we denote by $\mathrm{Aut}^{+}P$ the
group of $G$-bundle automorphism of $P$ preserving the orientation on $M$. We
also recall that if $\mathrm{Gau}^{\ast}P$ is the subgroup of gauge
transformations fixing a point of $P$, then $\mathrm{Gau}^{\ast}P$ acts freely
on $\mathcal{A}_{P}$ and $\mathcal{A}_{P}\rightarrow\mathcal{A}_{P}%
/\mathrm{Gau}^{\ast}P$ is a principal $\mathrm{Gau}^{\ast}P$-bundle (e.g. see
\cite{Donaldson}).

The principal $G$-bundle $\mathbb{P}=P\times\mathcal{A}_{P}\rightarrow
M\times\mathcal{A}_{P}$ has a tautological connection $\mathbb{A}\in\Omega
^{1}(P\times\mathcal{A}_{P},\mathfrak{g})$ defined by $\mathbb{A}%
_{(x,A)}(X,Y)=A_{x}(X)$ for $(x,A)\in P\times\mathcal{A}_{P}$, $X$ $\in
T_{x}P$, $Y\in T_{A}\mathcal{A}_{P}$. This connection is universal in the
sense that for any $A\in\mathcal{A}_{P}$ we have $A=t_{A}^{\ast}(\mathbb{A)}%
,$where $t_{A}\colon P\rightarrow\mathbb{P}$ is defined by $t_{A}(y)=(y,A)$
for any $y\in P$. We denote by $\mathbb{F}$ the curvature of $\mathbb{A}$. The
group $\mathrm{Aut}P$ acts on $\mathbb{P}$ by automorphisms and $\mathbb{A}$
is a $\mathrm{Aut}P$-invariant connection.

\begin{remark}
\label{Finite}As it is usual in Gauge theories (e.g. see \cite[Section
5.1.1.]{Donaldson}), in place of working in $\mathcal{A}_{P}$ and with the
group $\mathrm{Aut}^{+}P$, it is possible to formulate our results in terms of
families of connections. In our case, we need to consider a variation of this
concept, that we call equivariant families. Precisely, let $P\rightarrow M$ be
a principal $G$-bundle and let $\mathfrak{G}$ be a Lie group acting on a
manifold $T$, and also on $P\rightarrow M$ by automorphisms preseving the
orientation on $M$ (i.e., we have a Lie group homorphism $\rho\colon
\mathfrak{G}\rightarrow\mathrm{Aut}^{+}P$). A $\mathfrak{G}$-equivariant
family of connections parametrized by $T$ is a $\mathfrak{G}$-invariant
connection $B$ on the product $P\times T\rightarrow M\times T$. It defines a
$\rho$-equivariant map $b\colon T\rightarrow\mathcal{A}_{P}$, where
$b(t)=B|_{P\times\{t\}\rightarrow M\times\{t\}}$. All of the following results
and proofs are valid if we replace $(\mathrm{Aut}^{+}P,\mathcal{A}%
_{P},\mathbb{A)}$ by $(\mathfrak{G},T,B)$ (see also Remarks \ref{BundleConn}
and \ref{families final}).
\end{remark}

As the connection $\mathbb{A}$ is $\mathrm{Aut}P$-invariant, for any Weil
polynomial $p\in I^{r}(G)$ we can define the $\mathrm{Aut}P$-equivariant
characteristic form $p_{\mathrm{Aut}P}^{\mathbb{A}}\in\Omega_{\mathrm{Aut}%
P}^{2r}(M\times\mathcal{A}_{P})$ by $p_{\mathrm{Aut}P}^{\mathbb{A}%
}(X)=p(\mathbb{F}-v_{\mathbb{A}}(X))$ for $X\in\mathrm{aut}P$. If $M$ is a
compact oriented manifold of dimension $n$ without boundary the equivariant
form $p_{\mathrm{Aut}^{+}P}^{\mathbb{A}}\in\Omega_{\mathrm{Aut}^{+}P}%
^{2r}(M\times\mathcal{A}_{P})$ can be integrated over $M$ to obtain $\int
_{M}p_{\mathrm{Aut}^{+}P}^{\mathbb{A}}\in\Omega_{\mathrm{Aut}^{+}P}%
^{2r-n}(\mathcal{A}_{P})$. In particular, if $\dim M=2r-2$, we have
$\varpi_{P}^{p}=\int_{M}p_{\mathrm{Aut}^{+}P}^{\mathbb{A}}\in\Omega
_{\mathrm{Aut}^{+}P}^{2}(\mathcal{A}_{P})$ that can be written $\varpi_{P}%
^{p}=\omega_{P}^{p}+\mu_{P}^{p}$, with $\mu_{P}^{p}$ a comoment map for
$\omega_{P}^{p}$. The explicit expressions of these forms are
\begin{align}
(\omega_{P}^{p})_{A}(a,b)  & =r(r-1)%
%TCIMACRO{\tint \nolimits_{M}}%
%BeginExpansion
{\textstyle\int\nolimits_{M}}
%EndExpansion
p(a,b,F_{A},\overset{(r-2)}{\ldots},F_{A}),\nonumber\\
(\mu_{P}^{p})_{A}(X)  & =-r%
%TCIMACRO{\tint \nolimits_{M}}%
%BeginExpansion
{\textstyle\int\nolimits_{M}}
%EndExpansion
p(v_{A}(X),F_{A},\overset{(r-1)}{\ldots},F_{A})\label{exmu}%
\end{align}
for $A\in\mathcal{A}_{P}$, $a,b\in T_{A}\mathcal{A}_{P}\simeq\Omega
^{1}(M,\mathrm{ad}P)$ and $X\in\mathrm{aut}P$.

Let $A_{0}$ be a connection on $P\rightarrow M$ and let $\mathrm{pr}_{1}\colon
P\times\mathcal{A}_{P}\rightarrow P$ denote the projection. Then $\mathbb{A}$
and $\overline{A}_{0}=\mathrm{pr}_{1}^{\ast}A_{0}$ are connections on the same
bundle $P\times\mathcal{A}_{P}\rightarrow M\times\mathcal{A}_{P}$, and hence
we can define $Tp(\mathbb{A},\overline{A}_{0})\in\Omega^{2r-1}(M\times
\mathcal{A}_{P})$. We have $p(\mathbb{F})=dTp(\mathbb{A},\overline{A}%
_{0})+\mathrm{pr}_{1}^{\ast}p(F_{0})$. In particular, if $2r>n$ then
$p(\mathbb{F})=dTp(\mathbb{A},\overline{A}_{0})$ and hence $%
%TCIMACRO{\tint _{M}}%
%BeginExpansion
{\textstyle\int_{M}}
%EndExpansion
p(\mathbb{F})=d%
%TCIMACRO{\tint _{M}}%
%BeginExpansion
{\textstyle\int_{M}}
%EndExpansion
Tp(\mathbb{A},\overline{A}_{0})$.

\subsection{The bundle of connections}

The preceding constructions have a finite dimensional analog in terms of the
(finite dimensional) bundle of connections. We recall that given a principal
$G$-bundle $\pi\colon P\rightarrow M$, there exists a bundle $q\colon
C(P)\rightarrow M$ (called the bundle of connections) such that we have a
natural identification $\mathcal{A}_{P}\simeq\Gamma(M,C(P))$. For example we
can take $C(P)=(J^{1}P)/G$ where $J^{1}P$ is the first jet bundle of $P$. We
refer to \cite{geoconn} for more details on the geometry of $C(P)$. If
$A\in\mathcal{A}_{P}$, we denote by $\sigma_{A}$ the corresponding section of
$C(P)$. The pull-back bundle $\mathbf{P}=q^{\ast}P\rightarrow C(P)$ admits a
tautological connection defined by $\mathbf{A}_{(x,c)}(X,Y)=A_{x}(X) $ for
$(x,c)\in P\times C(P)$, $X$ $\in T_{x}P$, $Y\in T_{c}C(P)$ and where $A$ is
any connection such that $\sigma_{A}(x)=c$. This connection $\mathbf{A}$ has a
the following universal property: for any $A\in\mathcal{A}_{P}$ we have%
\begin{equation}
\bar{\sigma}_{A}^{\ast}(\mathbf{A})=A\label{universalA2}%
\end{equation}
where $\bar{\sigma}_{A}\colon P\rightarrow q^{\ast}P$ is defined by
$\bar{\sigma}_{A}(y)=(y,\sigma_{A}(\pi(y))$ for any $y\in P$. The\ group
$\mathrm{Aut}P$ acts on $C(P)$ in a natural way and the connection $\mathbf{A}
$ is $\mathrm{Aut}P$-invariant. Furthermore, the evaluation map $\mathrm{ev}%
\colon M\times\mathcal{A}_{P}\rightarrow C(P)\,$defined by $\mathrm{ev}%
(x,A)=s_{A}(x)$ is $\mathrm{Aut}P$-invariant$\ $and we have $\mathrm{ev}%
^{\ast}\mathbf{A}=\mathbb{A}$.

\subsection{Integrated equivariant Cheeger-Chern-Simons differential
characters\label{SectEquiCS}}

Let $G$ be a Lie group with a finite number of connected components, $\vec
{p}=(p,\Upsilon)$ a characteristic pair of degree $r$ and let $\pi\colon
P\rightarrow M$ be a principal $G$-bundle over a compact oriented manifold
without boundary of dimension $n=2r-2$.

If $(\phi,\gamma)\in\mathcal{C}^{\mathrm{Aut}^{+}P}(\mathcal{A}_{P})$ then
$\gamma$ can be extended to a curve $\tilde{\gamma}\colon\mathbb{R}%
\rightarrow\mathcal{A}_{P}$ by setting $\tilde{\gamma}(t)=(\phi_{\mathcal{A}%
_{P}})^{n}(\gamma(s))$ if $t=n+s$ for $n\in\mathbb{Z}$ and $s\in\lbrack0,1)$.
We define an action of $\mathbb{Z}$ on $P\times\mathbb{R}$ by setting
$n\cdot(y,t)=(\phi^{n}(y),t+n)$ for $n\in\mathbb{Z}$ and $(y,t)\in
P\times\mathbb{R}$, and a similar action on $M\times\mathbb{R}$. A connection
$A^{\gamma}\in\Omega^{1}(P\times\mathbb{R},\mathfrak{g})$ is defined by
$A^{\gamma}(X,h)=\tilde{\gamma}(t)(X)$ for $X\in TP$, $h\in T_{t}\mathbb{R}$,
$t\in\mathbb{R}$. Then $A^{\gamma}$ is a $\mathbb{Z}$-invariant connection
form on the principal $G$-bundle $P\times\mathbb{R}\rightarrow M\times
\mathbb{R}$. Hence $A^{\gamma}$ projects onto a connection $A_{\phi}^{\gamma}$
on the quotient bundle $(P\times\mathbb{R})/\mathbb{Z}\rightarrow
(M\times\mathbb{R)}/\mathbb{Z}$, and this bundle coincides with the mapping
torus bundle $P_{\phi}\rightarrow M_{\phi}$. We define the integrated
Cheeger-Chern-Simons equivariant differential character$\ $by setting
\begin{equation}
\Xi_{P}^{\vec{p}}(\phi,\gamma)=\mathrm{CS}^{\vec{p}}(A_{\phi}^{\gamma}%
)\in\mathbb{R}/\mathbb{Z}.\label{definition}%
\end{equation}

\begin{remark}
Strictly speaking, if $\gamma$ is smooth then $A^{\gamma}$ and $A_{\phi
}^{\gamma}$ are continuous but not differentiable in the $t$ direction. This
problem can be solved by considering a smooth non decreasing function
$\upsilon\colon I\rightarrow I$ such that $\upsilon$ has constant value $0$ in
$[0,\varepsilon]$ and $1$ on $[1-\varepsilon,1]$, and replacing $\gamma$ with
the reparametrization $\gamma\circ\upsilon$. The bundles with connection
$(P_{\phi},A_{\phi}^{\gamma\circ\upsilon})$ corresponding to different
$\upsilon$ are isomorphic and hence $\mathrm{CS}^{\vec{p}}(A_{\phi}%
^{\gamma\circ\upsilon})$ does not depend on the function $\upsilon$ chosen.

In a similar way, if $\gamma$ is only piecewise smooth, we can consider a
smooth reparametrization of $\gamma$ in order to obtain a smooth connection on
$P_{\phi}$.
\end{remark}

Next we prove that $\Xi_{P}^{\vec{p}}$ is an equivariant differential
character with equivariant curvature $\varpi_{P}^{p}$. To do it, we need to
consider a second equivalent definition of $\Xi_{P}^{\vec{p}}$.

\begin{remark}
\label{BundleConn}In the second definition we use the bundle of connections
$C(P)$ because it allows us to obtain the results by applying the
Cheeger-Chern-Simons constructions only for finite dimensional bundles. It is
also possible to obtain the same results by replacing the bundle
$\mathbf{P}\rightarrow C(P)$ with the infinite dimensional bundle
$P\times\mathcal{A}_{P}\rightarrow M\times\mathcal{A}_{P}$. However, in this
case, this requires the application of the generalization of the original
Cheeger-Chern-Simons construction to infinite dimensional bundles.
Furthermore, if we work with \emph{finite dimensional} equivariant families of
connections, then it is not necessary to use $C(P)$. In this case we can
replace $\mathbf{P}\rightarrow C(P)$\ by the bundle $P\times T\rightarrow
M\times T$, the connection $\mathbf{A}$ by the connection $B$ and the group
$\mathrm{Aut}^{+}P$ by $\mathfrak{G}$.
\end{remark}

As commented above, $\phi\in\mathrm{Aut}^{+}P$ defines an action of
$\mathbb{Z}$ on $P$ and also on $C(P)$ and $\mathbf{P}$. If $\mathrm{pr}%
\colon\mathbf{P}\times\mathbb{R}\rightarrow\mathbf{P}\,\ $is the projection,
the form $\mathrm{pr}^{\ast}\mathbf{A}$ is $\mathcal{\phi}$-invariant and
projects onto a connection $\mathbf{A}_{\phi}$ on the quotient bundle
$(\mathbf{P})_{\phi}\rightarrow(C(P))_{\phi}$. For any $\gamma\in
\mathcal{C}_{A}^{\phi}(\mathcal{A}_{P})$ we define $f_{\phi}^{\gamma}\colon
M_{\phi}\rightarrow(C(P))_{\phi}$ by $f_{\phi}^{\gamma}([(x,s)]_{\sim_{\phi}%
})=[(x,\sigma_{\gamma(s)}(x),s)]_{\sim_{\mathcal{\phi}}}$. It follows form
equation (\ref{universalA2}) that $(f_{\phi}^{\gamma})^{\ast}\mathbf{A}_{\phi
}=A_{\phi}^{\gamma}$ and by the naturality of the Cheeger-Chern-Simons
character (equation (\ref{naturality})) we obtain
\begin{equation}
\Xi_{P}^{\vec{p}}(\phi,\gamma)=\xi_{\mathbf{A}_{\phi}}^{\vec{p}}(f_{\phi
}^{\gamma}).\label{def2}%
\end{equation}

More generally, let $\mathcal{G}$ be a discrete group that acts on $P$ by
elements of $\mathrm{Aut}^{+}P\ $and let$\ E$ be a connected manifold in which
$\mathcal{G}$ acts freely. If $\mathrm{pr}\colon\mathbf{P}\times
E\rightarrow\mathbf{P}\,\ $is the projection, the form $\mathrm{pr}^{\ast
}\mathbf{A}$ is $\mathcal{G}$-invariant and projects onto a connection
$\mathbf{A}_{\mathcal{G}}^{E}$ on the quotient bundle $(\mathbf{P}\times
E)/\mathcal{G}\rightarrow(C(P)\times E)/\mathcal{G}$. Given $\phi
\in\mathcal{G}$, we choose a point $y\in E$ and a curve $\upsilon
\in\mathcal{C}_{y}^{\phi}(E)$. For any $\gamma\in\mathcal{C}_{A}^{\phi
}(\mathcal{A}_{P})$ we define $F_{\phi}^{\gamma,y,\upsilon}\colon M_{\phi
}\rightarrow(C(P)\times E)/\mathcal{G}$ by $F_{\phi}^{\gamma,y,\upsilon
}([(x,s)]_{\sim\phi})=[(x,\sigma_{\gamma(s)}(x),\upsilon(s))]_{\mathcal{G}}$.
Note that in the particular case in which $\mathcal{G}=\mathbb{Z}$, and
$E=\mathbb{R}$, $y=0$ and $\rho\in\mathcal{C}_{0}^{\phi}(\mathbb{R})$ is the
inclusion map $\rho\colon I\rightarrow\mathbb{R}$, we have $F_{\phi}%
^{\gamma,y,\upsilon}=f_{\phi}^{\gamma}$.

\begin{lemma}
\label{LemmaE}For any $\phi\in\mathcal{G}$, $\gamma\in\mathcal{C}_{A}^{\phi
}(\mathcal{A}_{P})$, $y\in E$ and $\upsilon\in\mathcal{C}_{y}^{\phi}(E)$ we
have $\Xi_{P}^{\vec{p}}(\phi,\gamma)=\xi_{\mathbf{A}_{\mathcal{G}}^{E}}%
^{\vec{p}}(F_{\phi}^{\gamma,y,\upsilon})$.
\end{lemma}

\begin{proof}
The\ element $\phi\in\mathcal{G}$ induces an action of $\mathbb{Z}$ on $E$. We
have natural maps
\[
(\mathbf{P}\times E\times\mathbb{R})/\mathbb{Z}%
\begin{array}
[c]{c}%
\overset{q_{E}}{\nearrow}\\
\underset{q^{\mathbb{R}}}{\searrow}%
\end{array}%
\begin{array}
[c]{c}%
(\mathbf{P}\times E)/\mathcal{G}\\
\\
(\mathbf{P}\times\mathbb{R})/\mathbb{Z}%
\end{array}
\]
and $q_{E}^{\ast}\mathbf{A}_{\mathcal{G}}^{E}=\mathbf{A}_{\mathbb{Z}}%
^{E\times\mathbb{R}}=q_{\mathbb{R}}^{\ast}\mathbf{A}$. Hence we have
$q_{E}^{\ast}(\xi_{\mathbf{A}_{\mathcal{G}}^{E}}^{\vec{p}})=\xi_{\mathbf{A}%
_{\mathbb{Z}}^{E\times\mathbb{R}}}^{\vec{p}}=q_{\mathbb{R}}^{\ast}%
(\xi_{\mathbf{A}_{\mathbb{Z}}^{\mathbb{R}}}^{\vec{p}})$. If $\rho\colon
I\rightarrow\mathbb{R}$ is the inclusion, then by applying equations
(\ref{naturality}) and (\ref{def2}) we obtain
\begin{align*}
\xi_{\mathbf{A}_{\mathcal{G}}^{E}}^{\vec{p}}(F_{\phi}^{\gamma,y,\upsilon})  &
=\xi_{\mathbf{A}_{\mathcal{G}}^{E}}^{\vec{p}}(q_{E}\circ F_{\phi}%
^{\gamma,(y,0),(\upsilon,\rho)})=\xi_{\mathbf{A}_{\mathbb{Z}}^{E\times
\mathbb{R}}}^{\vec{p}}(F_{\phi}^{\gamma,(y,0),(\upsilon,\rho)})\\
& =\xi_{\mathbf{A}_{\mathbb{Z}}^{\mathbb{R}}}^{\vec{p}}(q_{\mathbb{R}}\circ
F_{\phi}^{\gamma,(y,0),(\upsilon,\rho)})=\xi_{\mathbf{A}_{\mathbb{Z}%
}^{\mathbb{R}}}^{\vec{p}}(F_{\phi}^{\gamma,y,\rho})=\xi_{\mathbf{A}%
_{\mathbb{\phi}}}^{\vec{p}}(f_{\phi}^{\gamma})=\Xi_{P}^{\vec{p}}(\phi,\gamma).
\end{align*}

\end{proof}

\begin{proposition}
i) If $\phi,\phi^{\prime}\in\mathcal{G}$ and $\gamma\in\mathcal{C}^{\phi
}(\mathcal{A}_{P}),\gamma^{\prime}\in\mathcal{C}_{\gamma(1)}^{\phi^{\prime}%
}(\mathcal{A}_{P})$ then we have $\Xi_{P}^{\vec{p}}(\phi^{\prime}\cdot
\phi,\gamma^{\prime}\ast\gamma)=\Xi_{P}^{\vec{p}}(\phi^{\prime},\gamma
^{\prime})+\Xi_{P}^{\vec{p}}(\phi,\gamma)$.

ii) If $\zeta$ is a curve on $M$ such that $\zeta(0)=\gamma(0)$, and
$\gamma\in\mathcal{C}^{\phi}(M)$\ then $\overleftarrow{\zeta}\ast\gamma
\ast(\phi\cdot\zeta)\in\mathcal{C}^{\phi}(M)$ and $\Xi_{P}^{\vec{p}}%
(\phi,\overleftarrow{\zeta}\ast\gamma\ast(\phi\cdot\zeta))=\Xi_{P}^{\vec{p}%
}(\phi,\gamma)$.
\end{proposition}

\begin{proof}
i) Let $\mathcal{G}$ be the subgroup of $\mathrm{Aut}^{+}P$ generated by
$\phi$ and $\phi^{\prime}$ and $E$ a connected manifold in which $\mathcal{G}$
acts freely. We chose $y\in E$, $\upsilon\in\mathcal{C}_{y}^{\phi}(E)$,
$\upsilon^{\prime}\in\mathcal{C}_{y}^{\phi^{\prime}}(E)$ and we have
$F_{\phi^{\prime}\cdot\phi}^{\gamma^{\prime}\ast\gamma,y,\upsilon^{\prime}%
\ast\upsilon}=F_{\phi}^{\gamma,y,\upsilon}+F_{\phi^{\prime}}^{\gamma^{\prime
},y^{\prime},\upsilon^{\prime}}$ on $Z_{2k-1}((\mathbf{P}\times E)/\mathcal{G}%
)$, and the result follows from Lemma \ref{LemmaE}.

ii) We denote by $c_{0}$ the constant curve with value $0\in\mathbb{R}$ and we
define $\gamma_{1}=\overleftarrow{\zeta}\ast\gamma\ast(\phi\cdot\zeta)$ and
$\upsilon_{1}=c_{0}\ast\rho\ast(\phi\cdot c_{0})\in\mathcal{C}_{y}^{\phi}(E)$.
Then $F_{\phi}^{\gamma,0,\rho}=F_{\phi}^{\gamma_{1},0,\upsilon_{1}}$ on
$Z_{2k-1}((\mathbf{P})_{\phi})$ and the result follows.
\end{proof}

Next we compute the equivariant curvature of $\Xi_{P}^{\vec{p}}$ and we show
that the conditions iii) and iv) in the definition of equivariant differential
character are satisfied.

\begin{proposition}
If $\gamma\in\mathcal{C}^{e}(\mathcal{A}_{P})$ and $\gamma=\partial\nu$ for
$\nu\in C_{2}(M)$ then $\Xi_{P}^{\vec{p}}(\phi,\gamma)=\int_{\nu}\omega
_{P}^{\vec{p}}.$
\end{proposition}

\begin{proof}
As $\mathcal{A}_{P}$ is contractible, we can assume that $\nu$ is a map
$\nu\colon D^{2}\rightarrow M$, where $D^{2}\subset\mathbb{R}^{2}$ is a disk.
We recall that for $\phi=e$ the mapping torus is simply $M_{e}=M\times S^{1}$,
and we have a map $f_{e}^{\gamma}\colon M\times S^{1}\rightarrow C(P)\times
S^{1}$ such that $\Xi_{P}^{\vec{p}}(e,\gamma)=\xi_{\mathbf{A}_{\phi}}^{\vec
{p}}(f_{e}^{\gamma})$. We define $F^{\nu}\colon M\times D^{2}\rightarrow
C(P)\times D^{2}$ by $F^{\nu}(x,y)=(x,\sigma_{\nu(y)}(x),y)$ and we have
$\partial F^{\nu}=f_{e}^{\gamma}$. The connection $\mathbf{A}_{e}$ has an
obvious extension $\overline{\mathbf{A}}_{e}$\ to $\mathbf{P}\times D^{2}$ and
using equation (\ref{universalA2}) we have%
\[
\Xi_{P}^{\vec{p}}(\phi,\gamma)=\xi_{\mathbf{A}_{e}}^{\vec{p}}(f_{e}^{\gamma
})=\xi_{\overline{\mathbf{A}}_{e}}^{\vec{p}}(\partial F^{\nu})=%
%TCIMACRO{\tint \nolimits_{M\times D^{2}}}%
%BeginExpansion
{\textstyle\int\nolimits_{M\times D^{2}}}
%EndExpansion
(F^{\nu})^{\ast}p(\overline{\mathbf{F}}_{e})=%
%TCIMACRO{\tint \nolimits_{\nu}}%
%BeginExpansion
{\textstyle\int\nolimits_{\nu}}
%EndExpansion%
%TCIMACRO{\tint \nolimits_{M}}%
%BeginExpansion
{\textstyle\int\nolimits_{M}}
%EndExpansion
p(\mathbb{F}).
\]
We conclude that $\mathrm{curv}(\Xi_{P}^{\vec{p}})=\int_{M}p(\mathbb{F})$ and
that condition iii) is satisfied.
\end{proof}

Finally we prove condition vi).

\begin{proposition}
Let $X\in\mathrm{aut}P\,$and $A\in\mathcal{A}_{P}$. If $\nu_{t}^{x,X}%
(s)=\exp(tsX)\cdot A$ then we have $\left.  \frac{d}{dt}\right\vert _{t=0}%
\Xi_{P}^{\vec{p}}(\exp(tX),\nu_{t}^{x,X})=\mu_{P}^{p}(A)$.
\end{proposition}

\begin{proof}
The map $W_{t}^{M}\colon M\times\mathbb{R}\rightarrow M\times\mathbb{R}$,
$w_{t}^{M}(x,s)=(\phi_{st}(x),s)$ satisfies $W_{t}^{M}(e\cdot(x,s))=\phi
_{t}\cdot W_{t}^{M}(x,s)$ and hence it projects onto a diffeomorphism
$w_{t}^{M}\colon M\times S^{1}\rightarrow M_{\phi_{t}}$, and we have similar
maps for $\mathbf{P}$ and $C(P)$. The\ composition $(w_{t}^{C(P)})^{-1}\circ
f_{\phi}^{\gamma}\circ w_{t}^{M}$%
\[
M\times S^{1}\overset{w_{t}^{M}}{\longrightarrow}M_{\phi_{t}}\overset
{f_{\phi_{t}}^{\sigma_{t}}}{\longrightarrow}(C(P)\mathcal{)}_{\phi_{t}%
}\overset{(w_{t}^{C(P)})^{-1}}{\longrightarrow}C(P)\times S^{1}%
\]
is the $t$-independent map $f_{e}^{c_{A}}(x,s)=(x,\sigma_{A}(x),s)$. If
$\mathbf{B}_{t}=(w_{t}^{\mathbf{P}})^{\ast}\mathbf{A}_{\phi_{t}}$ then
$\Xi_{P}^{\vec{p}}(\phi_{t},\sigma_{t})=\xi_{\mathbf{B}_{t}}^{\vec{p}}%
(f_{e}^{c_{A}})$ and by Lemma \ref{LieCS} we obtain%
\begin{align}
\left.  \tfrac{d}{dt}\right\vert _{t=0}\Xi_{P}^{\vec{p}}(\phi_{t},\sigma_{t})
& =\left.  \tfrac{d}{dt}\right\vert _{t=0}\xi_{\mathbf{B}_{t}}^{\vec{p}}%
(f_{e}^{c_{A}})=r%
%TCIMACRO{\tint \nolimits_{M\times S^{1}}}%
%BeginExpansion
{\textstyle\int\nolimits_{M\times S^{1}}}
%EndExpansion
(f_{e}^{c_{A}})^{\ast}p(\mathbf{\dot{B}}_{0},\mathbb{F}_{e},\overset
{(r-1)}{\ldots},\mathbb{F}_{e})\nonumber\\
& =r%
%TCIMACRO{\tint \nolimits_{M\times S^{1}}}%
%BeginExpansion
{\textstyle\int\nolimits_{M\times S^{1}}}
%EndExpansion
p((f_{e}^{c_{A}})^{\ast}\mathbf{\dot{B}}_{0},F_{A},\overset{(r-1)}{\ldots
},F_{A}).\label{der}%
\end{align}
\bigskip

The connection $\mathbf{B}_{t}$ is the projection of the connection
$\mathbf{C}_{t}=(W_{t}^{\mathbf{P}})^{\ast}\mathrm{pr}^{\ast}\mathbf{A}$ to
$\mathbf{P}\times S^{1}$. Hence $\mathbf{\dot{B}}_{0}$ is the projection of
$\mathbf{\dot{C}}_{0}$. The vector field vector $Y\in\mathfrak{X}%
(\mathbf{P}\times\mathbb{R})$ given by\footnote{The minus sign appears by our
sign convention in the definition of the fundamental vector field}
$Y(y,s)=(-sX_{\mathbf{P}}(y),0)$ has $W_{t}^{\mathbf{P}}$\ as\ its flow. If we
define the vector $\overline{X}(y,s)=(X_{\mathbf{P}}(y),0)$ then by the
$\mathrm{Aut}^{+}P$-invariance of $\mathbf{A}$ we have $L_{\overline{X}%
}(\mathrm{pr}^{\ast}\mathbf{A})=0$ and
\begin{align*}
\mathbf{\dot{C}}_{0}  & =\left.  \tfrac{d}{dt}\right\vert _{t=0}%
(W_{t}^{\mathbf{P}})^{\ast}(\mathrm{pr}^{\ast}\mathbf{A)}=L_{Y}(\mathrm{pr}%
^{\ast}\mathbf{A)})=-sL_{\overline{X}}(\mathrm{pr}^{\ast}\mathbf{A)}%
)-(\mathbb{\iota}_{\overline{X}}\mathrm{pr}^{\ast}\mathbf{A)})ds\\
& =-\mathrm{pr}^{\ast}(v_{\mathbf{A}}(X))ds.
\end{align*}
Using equations (\ref{universalA2}) and (\ref{der}) we conclude that%
\begin{align*}
\left.  \frac{d}{dt}\right\vert _{t=0}\Xi_{P}^{\vec{p}}(\phi_{t},\sigma_{t})
& =-r%
%TCIMACRO{\tint \nolimits_{M\times S^{1}}}%
%BeginExpansion
{\textstyle\int\nolimits_{M\times S^{1}}}
%EndExpansion
p(v_{A}(X)ds,F,\overset{(r-1)}{\ldots},F)\\
& =-r%
%TCIMACRO{\tint \nolimits_{M}}%
%BeginExpansion
{\textstyle\int\nolimits_{M}}
%EndExpansion
p(v_{A}(X),F,\overset{(r-1)}{\ldots},F)\cdot%
%TCIMACRO{\tint \nolimits_{S^{1}}}%
%BeginExpansion
{\textstyle\int\nolimits_{S^{1}}}
%EndExpansion
ds\\
& =-r%
%TCIMACRO{\tint \nolimits_{M}}%
%BeginExpansion
{\textstyle\int\nolimits_{M}}
%EndExpansion
p(v_{A}(X),F,\overset{(r-1)}{\ldots},F)=\mu_{P}^{p}(X).
\end{align*}

\end{proof}

We conclude from the preceding results our main result:

\begin{theorem}
\label{final}$\Xi_{P}^{\vec{p}}$ is a $\mathrm{Aut}^{+}P$-equivariant
differential character on $\mathcal{A}_{P}$ with equivariant curvature
$\varpi_{P}^{p}$.
\end{theorem}

\begin{remark}
\label{families final}If we work with $\mathfrak{G}$-equivariant families of
connections as in Remark \ref{Finite}, then we obtain the following result:
For any $\mathfrak{G}$-equivariant family $(\mathfrak{G}$,$T,B)$ of
connections on $P\ $the map $\Xi_{P,B}^{\vec{p}}\colon\mathcal{C}%
^{\mathfrak{G}}(T)\rightarrow\mathbb{R}/\mathbb{Z}$ defined by $\Xi
_{P,B}^{\vec{p}}(\phi,\gamma)=\mathrm{CS}^{\vec{p}}(A_{\rho(\phi)}%
^{b\circ\gamma})$ (see Section \ref{SectEquiCS} for the notation) is a
$\mathfrak{G} $-equivariant differential character on $T$, i.e., $\Xi
_{P,B}^{\vec{p}}\in\hat{H}_{\mathfrak{G}}^{2}(T)$. In the particular case of
the family $(\mathrm{Aut}^{+}P$,$\mathcal{A}_{P},\mathbb{A})$ we obtain
Theorem \ref{final}. Conversely, given $\Xi_{P}^{\vec{p}}$, we have $\Xi
_{P,B}^{\vec{p}}=(b,\rho)^{\ast}\Xi_{P}^{\vec{p}}$. Hence this result is an
equivalent formulation of Theorem \ref{final} that does not involve the
infinite dimensional strucutures of $\mathcal{A}_{P}$ and $\mathrm{Aut}^{+}P$.
\end{remark}

We can define the Chern-Simons bundle in terms of the equivariant differential
character $\Xi_{P}^{\vec{p}}$ by choosing a background connection $A_{0}%
\in\mathcal{A}_{P}$ and by applying Theorem \ref{PropAction} to the form
$\lambda=\int_{M}Tp(\mathbb{A},\overline{A}_{0})\in\Omega^{1}(\mathcal{A}%
_{P})$. Precisely, the Chern-Simons bundle is the $\mathrm{Aut}^{+}%
P$-equivariant $U(1)$-bundle given by the trivial bundle $\mathcal{A}%
_{P}\times U(1)\rightarrow\mathcal{A}_{P}$ with the action defined by the
cocycle $\alpha\colon\mathrm{Aut}^{+}P\times\mathcal{A}_{P}\rightarrow
\mathbb{R}/\mathbb{Z}$ where $\alpha_{\phi}(A)=\int_{\gamma}\lambda-\Xi
_{P}^{\vec{p}}(\gamma)$ for any $\gamma\in\mathcal{C}_{A}^{\phi}(M)$. We show
in the next section that in the case in which $P$ is a trivial bundle over a
2-manifold,\ this definition coincides with the usual definition of the
Chern-Simons bundle.

\subsection{The equivariant holonomy of the Chern-Simons bundle for trivial
bundles\label{SectEquiHolCS}}

We recall the construction given in \cite{RSW} of the Chern-Simons bundle for
a trivial principal $G$-bundle $P=M\times G\rightarrow M$ over a compact
oriented 2-manifold $M$\ without boundary. In \cite{RSW} it is considered the
case of the group $G=SU(2)\,$but the construction is valid for any trivial
bundle. If $p\in I_{\mathbb{Z}}^{2}(G)$, we define a cocycle $\alpha$ on
$\mathcal{A}_{P}$ for the group $\mathcal{G}=\mathrm{Gau}P\simeq
\mathcal{C}^{\infty}(M,G)$ in the following way. Let $\widetilde{M}$ be a
compact $3$-manifold with boundary $\partial\widetilde{M}=M$ and let
$\widetilde{P}=\widetilde{M}\times G$. We denote by $A_{0}$ and $\widetilde
{A}_{0}$ the connections associated to the product structure on $M\times G$
and $\widetilde{M}\times G$ respectively.

We define a cocycle $\alpha\colon\mathcal{G}\times$ $\mathcal{A}%
_{P}\rightarrow\mathbb{R}/\mathbb{Z}$ by setting for $A\in\mathcal{A}_{P}$ and
$\phi\in\mathcal{G}$
\[
\alpha_{\phi}(A)=\mathrm{CS}^{p}(\widetilde{\phi}(\widetilde{A}))-\mathrm{CS}%
^{p}(\widetilde{A})\operatorname{mod}\mathbb{Z},
\]
where $\widetilde{A}\in\mathcal{A}_{\widetilde{P}}$ and $\widetilde{\phi}%
\in\mathrm{Gau}\widetilde{P}$ are extensions of $A$ and $\phi$ to
$\widetilde{P}$. It is easily seen (see \cite{RSW}) that the condition $p\in
I_{\mathbb{Z}}^{2}(G)$ implies that $\alpha_{\phi}(A)\operatorname{mod}%
\mathbb{Z}$ is independent of the extensions $\widetilde{A}$, $\widetilde
{\phi}$\ chosen, and that $\alpha$ satisfies the cocycle condition. Hence
$\alpha$ defines a $\mathcal{G}$-equivariant $U(1)$-bundle $\mathcal{U}%
^{p}\rightarrow\mathcal{A}_{P}$. Furthermore, the form $\lambda=\int
_{M}Tp(\mathbb{A},\overline{A}_{0})\in\Omega^{1}(\mathcal{A}_{P})$ (see
Section \ref{SectGeoConn})\ determines a $\mathcal{G}$-invariant connection
$\Theta_{P}^{p}=\vartheta+2\pi i\lambda$ on $\mathcal{U}^{p}\rightarrow
\mathcal{A}_{P}$. The $\mathcal{G}$-equivariant $U(1)$-bundle with connection
$(\mathcal{U}_{P}^{p},\Theta_{P}^{p}$) is called the Chern-Simons bundle of $p
$. Next we compute the $\mathcal{G}$-equivariant holonomy of $\Theta_{P}^{p} $
and we show that it coincides (up to a sign) with the $\mathcal{G}
$-equivariant character $\Xi_{P}^{\vec{p}}$.

\begin{proposition}
If $\phi\in\mathcal{G}$ and $\gamma\in\mathcal{C}^{\phi}(\mathcal{A}_{P})$
then we have $\mathrm{hol}_{\phi}^{\Theta_{P}^{p}}(\gamma)=-\mathrm{CS}%
^{p}(A_{\phi}^{\gamma})$.
\end{proposition}

\begin{proof}
We denote by $\widetilde{\mathbb{A}}$ the tautological connection on
$\widetilde{M}\times\mathcal{A}_{\widetilde{P}}$, by $\widetilde{\mathbb{F}}$
its curvature, we set $\widetilde{\mathbb{\mathcal{G}}}=\mathrm{Gau}%
\widetilde{P}$ and we denote by $r\colon\mathcal{A}_{\widetilde{P}}%
\rightarrow\mathcal{A}_{P}$ the restriction map. Given $\phi\in\mathcal{G}$
and $\gamma\in\mathcal{C}_{A}^{\phi}(\mathcal{A}_{P})$, we can find extensions
$\widetilde{\phi}\in\widetilde{\mathbb{\mathcal{G}}}$, $\widetilde{A}%
\in\mathcal{A}_{\widetilde{P}}$ of $\phi$ and $A$. We consider the mapping
tori bundles $P_{\phi}\rightarrow M_{\phi}$, $\widetilde{P}_{\widetilde{\phi}%
}\rightarrow\widetilde{M}_{\widetilde{\phi}}$ and we have $P_{\phi}%
=\partial\widetilde{P}_{\widetilde{\phi}}$. We choose an extension
$A_{\widetilde{\phi}}^{\widetilde{\gamma}}$ of $A_{\phi}^{\gamma}$ to
$\widetilde{P}_{\widetilde{\phi}}$,\ that corresponds to a curve
$\widetilde{\gamma}\in\mathcal{C}_{\widetilde{A}}^{\widetilde{\phi}%
}(\mathcal{A}_{\widetilde{P}})$.$\,\ $If $\lambda=\int_{M}Tp(\mathbb{A}%
,\overline{A}_{0})\in\Omega^{1}(\mathcal{A}_{P})$ and $\beta=\int
_{\widetilde{M}}Tp(\widetilde{\mathbb{A}},\overline{\widetilde{A}}_{0}%
)\in\Omega^{0}(\mathcal{A}_{P})$ by Stokes Theorem we have $d\beta
=\int_{\widetilde{M}}p(\widetilde{\mathbb{F}})-\int_{M}Tp(\widetilde
{\mathbb{A}},\overline{\widetilde{A}}_{0})=\int_{\widetilde{M}}p(\widetilde
{\mathbb{F}})-r^{\ast}\lambda$. By applying Proposition \ref{holAlfa} we
obtain
\begin{align*}
\mathrm{hol}_{\phi}^{\Theta}(\gamma)  & =-%
%TCIMACRO{\tint \nolimits_{\gamma}}%
%BeginExpansion
{\textstyle\int\nolimits_{\gamma}}
%EndExpansion
\lambda-\alpha_{\phi}(x)\\
& =-%
%TCIMACRO{\tint \nolimits_{\gamma}}%
%BeginExpansion
{\textstyle\int\nolimits_{\gamma}}
%EndExpansion
\lambda-%
%TCIMACRO{\tint \nolimits_{\widetilde{M}}}%
%BeginExpansion
{\textstyle\int\nolimits_{\widetilde{M}}}
%EndExpansion
Tp(\widetilde{\gamma}(1),B_{0})+%
%TCIMACRO{\tint \nolimits_{\widetilde{M}}}%
%BeginExpansion
{\textstyle\int\nolimits_{\widetilde{M}}}
%EndExpansion
Tp(\widetilde{\gamma}(0),B_{0})\\
& =-%
%TCIMACRO{\tint \nolimits_{\widetilde{\gamma}}}%
%BeginExpansion
{\textstyle\int\nolimits_{\widetilde{\gamma}}}
%EndExpansion
r^{\ast}\lambda-\beta(\widetilde{\gamma}(1))+\beta(\widetilde{\gamma}(0))=-%
%TCIMACRO{\tint \nolimits_{\widetilde{\gamma}}}%
%BeginExpansion
{\textstyle\int\nolimits_{\widetilde{\gamma}}}
%EndExpansion
(r^{\ast}\lambda+d\beta)\\
& =-%
%TCIMACRO{\tint \nolimits_{\widetilde{\gamma}}}%
%BeginExpansion
{\textstyle\int\nolimits_{\widetilde{\gamma}}}
%EndExpansion%
%TCIMACRO{\tint \nolimits_{\widetilde{M}}}%
%BeginExpansion
{\textstyle\int\nolimits_{\widetilde{M}}}
%EndExpansion
p(\widetilde{\mathbb{F}})=-%
%TCIMACRO{\tint \nolimits_{\widetilde{M}\times I}}%
%BeginExpansion
{\textstyle\int\nolimits_{\widetilde{M}\times I}}
%EndExpansion
p(F_{A^{\widetilde{\gamma}}})=-%
%TCIMACRO{\tint \nolimits_{\widetilde{M}_{\widetilde{\phi}}}}%
%BeginExpansion
{\textstyle\int\nolimits_{\widetilde{M}_{\widetilde{\phi}}}}
%EndExpansion
p(F_{A_{\widetilde{\phi}}^{\widetilde{\gamma}}})\\
& =-\xi_{A_{\phi}^{\gamma}}^{\vec{p}}(M_{\phi})=-\mathrm{CS}(A_{\phi}^{\gamma
}).
\end{align*}

\end{proof}

In the case of $G=SU(2)$ considered in \cite{RSW} any principal $SU(2)$-bundle
over a manifold of dimension $2$ or $3$ is trivializable, and we can apply the
preceding construction to define the Chern-Simons line bundle. For other
groups (for example $G=U(1)$) there are nontrivial principal $G$-bundles and
this construction cannot\ applied. However, our construction in Section
\ref{SectEquiCS}\ can be applied in this case. Furthermore, our
construction\ is valid in any even dimension $m=2k-2$, for arbitrary group $G
$ and $\Xi_{P}^{\vec{p}}$\ is equivariant with respect of the action of the
group $\mathrm{Aut}^{+}P$ (and not only for gauge transformations).

Finally we relate the character $\Xi_{P}^{\vec{p}}$ with the bundles defined
in \cite{CSconnections}. If $\mathrm{pr}\colon P\times\mathcal{A}_{P}%
\times\mathbb{R}\rightarrow P\times\mathcal{A}_{P}\,\ $is the projection, for
any $\phi\in\mathrm{Aut}^{+}P$ the form $\mathrm{pr}^{\ast}\mathbb{A}_{\phi}$
is $\phi$-invariant and projects onto a connection $\mathbb{A}_{\phi}$ on
$(P\times\mathcal{A}_{P})_{\phi}\rightarrow(M\times\mathcal{A}_{P})_{\phi}$.
The differential character $\xi_{\mathbb{A}_{\phi}}^{\vec{p}}\in\hat{H}%
^{2k}((M\times\mathcal{A}_{P})_{\phi})$ can be integrated over $M$ and we
obtain a differential character $\int_{M}\xi_{\mathbb{A}_{\phi}}^{\vec{p}}%
\in\hat{H}^{2}((\mathcal{A}_{P})_{\phi})$. If $\gamma\in\mathcal{C}_{A}^{\phi
}(\mathcal{A}_{P})$ then we can define a curve $\gamma_{\phi}$ on
$(\mathcal{A}_{P})_{\phi}$ by setting $\gamma_{\phi}(t)=[\gamma(t),t]_{\sim
\phi}$. We define $\Lambda_{P}^{\vec{p}}(\phi,\gamma)=\left(  \int_{M}%
\xi_{\mathbb{A}_{\phi}}^{\vec{p}}\right)  (\gamma_{\phi})$. It is shown in
\cite{CSconnections}\ that if $A_{0}$ is a connection on $P\rightarrow M$,
then the form $\lambda=%
%TCIMACRO{\tint _{M}}%
%BeginExpansion
{\textstyle\int_{M}}
%EndExpansion
Tp(\mathbb{A},\overline{A}_{0})\in\Omega^{1}(\mathcal{A}_{P})$ satisfies
$d\lambda=\omega_{P}^{p}$, the map $\beta_{\phi}(A)=\int_{\gamma}%
\lambda-\Lambda_{P}^{\vec{p}}(\phi,\gamma)$ for $\phi\in\mathrm{Aut}^{+}P$,
$\gamma\in\mathcal{C}_{A}^{\phi}(\mathcal{A}_{P})$ satisfies the cocycle
condition and the connection $\Theta^{\vec{p}}=\vartheta-2\pi i\lambda$ is
invariant under the action of $\mathrm{Aut}^{+}P$ on $\mathcal{A}_{P}\times
U(1)$ induced by the cocycle $\alpha$. It follows from Proposition
\ref{holAlfa}\ that $\Lambda_{P}^{\vec{p}}=\mathrm{hol}_{\mathrm{Aut}^{+}%
P}^{\Theta^{\vec{p}}}$. It can be proved using the definition of the fiber
integral of differential characters \ and equation\ (\ref{def2})\ that
$\Xi_{P}^{\vec{p}}=\Lambda_{P}^{\vec{p}}=\mathrm{hol}_{\mathrm{Aut}^{+}%
P}^{\Theta^{\vec{p}}}$. This result provides an alternative proof of the fact
that $\Xi_{P}^{\vec{p}}\in\hat{H}_{\mathrm{Aut}^{+}P}^{2}(\mathcal{A}_{P})$,
but it needs to use fiber integration of differential characters and also the
Cheeger-Chern-Simons construction applied to infinite dimensional bundles.

In the rest of the paper we show how the results of \cite{CSconnections} can
be obtained form the equivariant differential character $\Xi_{P}^{\vec{p}}$.

\begin{example}
Let $M$ be a Riemann surface, $P=M\times SU(2)$ the trivial principal
$SU(2)$-bundle and $\vec{p}$ the characteristic pair corresponding to the
second Chern class. In this case $\varpi_{P}^{p}$ coincide with the
Atiyah-Bott symplectic structure $\omega\in\Omega^{2}(\mathcal{A}_{P})$ and
moment map given by $\omega_{P}^{p}(a,b)=-\frac{1}{4\pi^{2}}\int
_{M}\mathrm{tr}(a\wedge b)$ and $(\mu_{P}^{p})_{X}(A)=\frac{1}{4\pi^{2}}%
\int_{M}\mathrm{tr}(v_{A}(X)\wedge F)$, for $A\in\mathcal{A}_{P}$,
$a,b\in\Omega^{1}(\Sigma,\mathrm{ad}P)\simeq T_{A}(\mathcal{A}_{P})$ and
$X\in\mathrm{aut}P\,$. We have a $\mathrm{Aut}^{+}P$-equivariant differential
character\ $\Xi_{P}^{\vec{p}}\in\hat{H}_{\mathrm{Aut}^{+}P}^{2}(\mathcal{A}%
_{P})$ that by Corollary \ref{ExistenceBundle}\ determines (up to an
isomorphism) a $\mathrm{Aut}^{+}P$-equivariant $U(1)$-bundle $\mathcal{U}%
^{\vec{p}}\rightarrow\mathcal{A}_{P}$ with connection $\Theta^{\vec{p}}.$

If $\mathcal{F}_{P}\subset\mathcal{A}_{P}$ is the space of flat connections,
then we have $\mathcal{F}_{P}\subset(\mu_{P}^{p})^{-1}(0)$, and by Proposition
\ref{quot2}\ $\Xi_{P}^{\vec{p}}$ projects onto a $\mathrm{Aut}^{+}%
P$-equivariant differential character $\underline{\Xi}_{P}^{\vec
{p},\mathcal{F}}\in\hat{H}_{\mathrm{Aut}^{+}P}^{2}(\mathcal{F}_{P}%
/\mathrm{Gau}^{\ast}P)$ on the moduli space of flat connections. The
connection $\Theta^{\vec{p}}$ projects to the quotient $\mathcal{U}^{\vec{p}%
}/\mathrm{Gau}^{\ast}P\rightarrow\mathcal{F}_{P}/\mathrm{Gau}^{\ast}P$\ and
this bundle is isomorphic to Quillen's determinant line bundle (see
\cite{CSconnections}).
\end{example}

\subsection{Action by Gauge transformations\label{SectGau}}

Now we consider that $M$ is an arbitrary oriented manifold, and $P\rightarrow
M$ a principal $G$-bundle. Let $C$ be a compact oriented manifold of dimension
$2r-2$ and let $c\colon C\rightarrow M$ be a smooth map (for example $C$ can
be a submanifold of $M$). Then we have a group homomorphism $\rho
\colon\mathrm{Gau}P\rightarrow\mathrm{Aut}^{+}(c^{\ast}P)$\ and a $\rho
$-equivariant map $f\colon\mathcal{A}_{P}\rightarrow\mathcal{A}_{c^{\ast}P}$.
By Proposition \ref{pullback} we have an equivariant differential character
$(f,\rho)^{\ast}\Xi_{c^{\ast}P}^{\vec{p}}\in\hat{H}_{\mathrm{Gau}P}%
^{2}(\mathcal{A}_{P})$.

\subsection{Manifolds with boundary}

Let $M$ be an oriented manifold with compact boundary $\partial M$. We denote
by $\upsilon\colon\partial M\rightarrow M$ the inclusion map. If $P\rightarrow
M$ is a principal $G$-bundle, then we denote by $\partial P$ the bundle
$\upsilon^{\ast}P\rightarrow\partial M$. We have a group homomorphism
$\rho\colon\mathrm{Aut}^{+}P\rightarrow\mathrm{Aut}^{+}\partial P$ and a
$\rho$-equivariant map $f\colon\mathcal{A}_{P}\rightarrow\mathcal{A}_{\partial
P}$. If $\vec{p}=(p,\Upsilon)$ is a characteristic pair of degree $r$ and
$\dim M=2r-1$, then we have the Chern-Simons character $\Xi_{\partial P}%
^{\vec{p}}\in\hat{H}^{2}(\mathcal{A}_{\partial P})$. By Proposition
\ref{pullback}\ these data determine a differential character $(f,\rho)^{\ast
}\Xi_{\partial P}^{\vec{p}}\in\hat{H}_{\mathrm{Aut}^{+}P}^{2}(\mathcal{A}%
_{P})$. As commented in the Introduction, the equivariant bundles associated
to the character $\Xi_{\partial P}^{\vec{p}}$ are called the Chern-Simons
bundles because the Chern-Simons action on $M$ determines a $\mathrm{Aut}%
^{+}P$-equivariant section of $(f,\rho)^{\ast}\Xi_{\partial P}^{\vec{p}}$ (see
\cite{CSconnections}). This fact depends on the background connection $A_{0}$
chosen in the definition of the bundle and of the Chern-Simons action. We
present an intrinsic version of this result.

If $q\in I^{r}(M)$ is a Weil polynomial of degree $r$ and $M$ is a manifold
with compact boundary of dimension $2r-1$, then $\int_{M}q(\mathbb{F}%
)\in\Omega^{1}(\mathcal{A}_{P})^{\mathrm{Aut}^{+}P}$.

\begin{proposition}
If $q\in I^{r}(M)$ is a Weil polynomial of degree $r$ and $M$ is a compact
oriented manifold of dimension $2r-1$, then $\int_{M}q(\mathbb{F})\in
\Omega^{1}(\mathcal{A}_{P})^{\mathrm{Aut}^{+}P}$ and $\varsigma(\int
_{M}q(\mathbb{F}))(\phi,\gamma)=\int_{M_{\phi}}q(F_{A_{\phi}^{\gamma}})$.
\end{proposition}

\begin{proof}
For any $(\phi,\gamma)\in\mathcal{C}^{\mathrm{Aut}^{+}P}(\mathcal{A}_{P})$ we
have $\varsigma(\int_{M}q(\mathbb{F}))(\phi,\gamma)=\int_{\gamma}\int
_{M}q(\mathbb{F})=\int_{M\times\gamma}q(\mathbb{F})=\int_{M\times
I}q(F_{A^{\gamma}})=\int_{M_{\phi}}q(F_{A_{\phi}^{\gamma}}).$
\end{proof}

\begin{proposition}
\label{CS}We have $(f,\rho)^{\ast}\Xi_{\partial P}^{\vec{p}}=\varsigma
(\int_{M}p(\mathbb{F}))$.
\end{proposition}

\begin{proof}
Let us fix $(\phi,\gamma)\in\mathcal{C}^{\mathrm{Aut}^{+}P}(\mathcal{A}_{P})$
and define $\chi=(f,\rho)^{\ast}\Xi_{\partial P}^{\vec{p}}$. The curve
$\gamma\in\mathcal{C}^{\phi}(\mathcal{A}_{P})$ induces a curve $\partial
\gamma\in\mathcal{C}^{\phi}(\mathcal{A}_{\partial P})$. If $P_{\phi
}\rightarrow M_{\phi}$ is the mapping torus bundle of $P$, then we have
$\partial(P_{\phi})=(\partial P)_{\phi}$. Furthermore, by the properties of
the Cheeger-Chern-Simons characters and the preceding proposition we have
$\chi(\phi,\gamma)=\xi_{A_{\phi}^{\partial\gamma}}^{\vec{p}}((\partial
M)_{\phi})=\xi_{A_{\phi}^{\gamma}}^{\vec{p}}(\partial(M_{\phi}))=\int
_{M_{\phi}}p(F_{A_{\phi}^{\gamma}})=\varsigma(\int_{M}p(\mathbb{F}%
))(\phi,\gamma)$.
\end{proof}

In particular, it follows from Proposition \ref{anomalyC} that the
$\mathrm{Aut}^{+}P$-equivariant $U(1)$-bundle associated to $(f,\rho)^{\ast
}\Xi_{\partial P}^{\vec{p}}$ is trivial and hence it admits a $\mathrm{Aut}%
^{+}P$-invariant section. We hope that our approach using equivariant
differential characters could be used to study Chern-Simons theory for
arbitrary bundles and groups.

\subsection{Riemannian metrics and diffeomorphisms\label{Sec Metrics}}

Let $\vec{p}=(p,\Upsilon)$ a characteristic pair of degree $2k$ for the group
$Gl(4k-2,\mathbb{R})$. For example we can consider the pair corresponding to
the $k$-th Pontryagin class. Let $M$ be a compact oriented manifold without
boundary of dimension $n=4k-2$. We denote by $FM\rightarrow M$ the frame
bundle of $M$, by $\mathcal{M}_{M}$ the space of Riemannian metrics on $M$ and
by $\mathcal{D}_{M}^{+}$ the group of orientation preserving diffeomorphism of
$M$. As $\mathcal{D}_{M}^{+}$ acts in a natural way on $F(M)$ by
automorphisms, we have a natural homomorphism $\rho\colon\mathcal{D}_{M}%
^{+}\rightarrow\mathrm{Aut}^{+}F(M)$. The Levi-Civita map $LC\colon
\mathcal{M}_{M}\rightarrow\mathcal{A}_{FM}$ is $\rho$-equivariant and hence by
Proposition \ref{pullback} we have an equivariant differential character
$\Sigma_{\mathcal{M}_{M}}^{\vec{p}}=(LC,\rho)^{\ast}\Xi_{P}^{\vec{p}}\in
\hat{H}_{\mathcal{D}_{M}^{+}}^{2}(\mathcal{M}_{M})$ with curvature
$\sigma_{\mathcal{M}_{M}}^{p}=(LC,\rho)^{\ast}\varpi_{FM}^{p}$, that can be
written $\sigma_{\mathcal{M}_{M}}^{p}=\omega_{\mathcal{M}_{M}}^{p}%
+\mu_{\mathcal{M}_{M}}^{p}$.

We consider in more detail the case $k=1$. If $M$ is a Riemann surface of
genus $g>1$, and $\mathcal{M}_{M}^{-1}$ is the space of metrics of constant
curvature $-1$, then we have $\mathcal{M}_{M}^{-1}\subset(\mu_{\mathcal{M}%
_{M}}^{p})^{-1}(0)$ (see \cite{WP}). If $\mathcal{D}_{M}^{0}$\ denotes the
connected component with the identity on $\mathcal{D}_{M}^{+}$, then
$\mathcal{D}_{M}^{0}$ acts freely on $\mathcal{M}_{M}^{-1}$ and the
Teichm\"{u}ller space of $M$ is defined by $\mathcal{T}_{M}=\mathcal{M}%
_{M}^{-1}/\mathcal{D}_{M}^{0}$. As it is well known (e.g. see \cite{Tromba}),
$\mathcal{T}(M)$\ is a contractible manifold of real dimension $6g-6$.
Furthermore, it is proved in \cite{WP} that the form obtained on
$\mathcal{T}_{M}$ from $\omega_{\mathcal{M}_{M}}^{p}$ by symplectic reduction
is $\frac{1}{2\pi^{2}}\sigma_{\mathrm{WP}}$, where $\sigma_{\mathrm{WP}}$ is
the symplectic form of the Weil-Petersson metric on $\mathcal{T}_{M}$. By
Remark \ref{quot20} we obtain an equivariant differential character
$\underline{\Sigma}_{\mathcal{M}_{M}}^{\vec{p}}\in\hat{H}_{\Gamma_{M}}%
^{2}(\mathcal{T}_{M})$ with curvature $\frac{1}{2\pi^{2}}\sigma_{\mathrm{WP}}%
$, where $\Gamma_{M}=\mathcal{D}_{M}^{+}\mathbf{/}\mathcal{D}_{M}^{0}$ is the
mapping class group of $M$. By Corollary \ref{ExistenceBundle}\ $\underline
{\Sigma}_{\mathcal{M}_{M}}^{\vec{p}}$ determines (up to an isomorphism) a
$\Gamma_{M}$-equivariant $U(1)$-bundle with connection over $\mathcal{T}_{M}$.

\subsubsection{Manifolds with boundary}

Let $M$ be an oriented manifold of dimension $n=4k-1$ with compact boundary
$\partial M$. We have a homomorphism $\rho\colon\mathcal{D}_{M}^{+}%
\rightarrow\mathrm{Aut}^{+}F(\partial M)$ and a $\rho$-equivariant map
$f\colon\mathcal{M}_{M}\rightarrow\mathcal{A}_{F(\partial M)}$. \ By
Proposition \ref{pullback} we obtain an equivariant differential character
$(f,\rho)^{\ast}\Xi_{F(\partial M)}^{\vec{p}}\in\hat{H}_{\mathcal{D}_{M}^{+}%
}^{2}(\mathcal{M}_{M})$, and by Corollary \ref{ExistenceBundle} a
$\mathcal{D}_{M}^{+}$-equivariant $U(1)$-bundle with connection over
$\mathcal{M}_{M}$. We also have a result analogous to Proposition \ref{CS}.

\end{document}